\documentclass[preprint]{imsart}

\RequirePackage[OT1]{fontenc}
\RequirePackage{amsthm,amsmath}
\RequirePackage[numbers]{natbib}
\RequirePackage[colorlinks,citecolor=blue,urlcolor=blue]{hyperref}
\RequirePackage{mathrsfs} 

\RequirePackage{graphicx}

\RequirePackage{mathptmx}      % use Times fonts if available on your TeX system

\RequirePackage{latexsym}

\RequirePackage{amsmath, amssymb,amsfonts}

\RequirePackage{natbib}

\RequirePackage{color}

\RequirePackage{url}

\RequirePackage{algorithm}% http://ctan.org/pkg/algorithms
\RequirePackage{algpseudocode}% http://ctan.org/pkg/algorithmicx
\algnewcommand\algorithmicinput{\textbf{Input:}}
\algnewcommand\INPUT{\item[\algorithmicinput]}

\theoremstyle{plain}
\newtheorem{theorem}{Theorem}

\newtheorem{lemma}{Lemma}

\newtheorem{assumption}{Assumption}

\arxiv{arXiv:0000.0000}

\startlocaldefs
\numberwithin{equation}{section}
\endlocaldefs

\begin{document}

\begin{frontmatter}
\title{Accelerate iterated filtering}
\runtitle{Accelerate iterated filtering}

\begin{aug}
\author{\fnms{Dao} \snm{Nguyen}\ead[label=e1]{dxnguyen@olemiss.edu}}

\runauthor{D. Nguyen}

\affiliation{Some University\thanksmark{m1} and Another University\thanksmark{m2}}

\address{Departments of Mathematics,\\
University of Mississippi, Oxford, Mississippi, USA\\
\printead{e1}}

\end{aug}

\begin{abstract}
In simulation-based inferences for partially observed Markov process models (POMP), the by-product of the Monte Carlo 
filtering is an approximation of the log likelihood function. Recently, iterated filtering \cite{ionides06-pnas,ionides11} has originally been
introduced and it has been shown that the gradient of the log likelihood can also be approximated. 
Consequently, different stochastic optimization algorithm can be applied to estimate the parameters of the underlying models. 
As accelerated gradient is an efficient approach in the optimization literature, we show that we 
can accelerate iterated filtering in the same manner and inherit that high convergence rate while relaxing the restricted  
conditions of unbiased gradient approximation.  
We show that this novel algorithm can be applied to both convex and nonconvex log likelihood functions. 
In addition, this approach has substantially outperformed most of other previous approaches in a toy example and in a challenging scientific problem of modeling 
infectious diseases.  
\end{abstract}


\end{frontmatter}

\section{Introduction}

The last decade has seen a great increase in the use of simulation-based
inference where numerical approximations are based on either Markov
chain Monte Carlo or sequential Monte Carlo sampling. These approaches
have become popularized, in part, because of the increasing computational
power and the emergence of efficient stochastic optimization algorithms.
On the Bayesian paradigm, particle Markov chain Monte Carlo has been
introduced and popularized by Doucet and collaborators \cite{andrieu10,andrieu2015convergence,pitt2012some}.
Similar ideas have been developed previously \cite{leggetter1995maximum,doucet02,gaetan03,jacquier07}
but in different contexts than simulation-based inferences. On the
frequentist paradigm, \cite{ionides06-pnas,ionides11} have introduced an original
approach to perform simulation-based parameter inference in POMP models
by combining stochastic gradient approximation and particle filtering.
In this paper, we will focus on improving one of the most popular
algorithm of this class, namely, iterated filtering (IF). Iterated
filtering uses an approximation estimate of the gradient of the log likelihood
computed from particle filters while proposing an artificial perturbation
moves to update the parameters. This class of algorithm is attractive
because it enables routine simulation-based parameter inferences in
general POMP model, even in the cases of intractable likelihoods.
Due to some interesting theoretical properties \cite{ionides11,ionides15,nguyen2016another},
its applications range in various fields such as biology, ecology,
economics and engineering \cite{lindstrom12,lindstrom2013tuned,laneri10,bhadra10,breto11,breto2014idiosyncratic}.

Iterated filtering was later theoretically developed by \citet{ionides11}.
Recently, \citet{lindstrom12} extended it to improve on numerical
performance while \citet{doucet2013derivative} expanded it to include
filtering/smoothing with quite attractive theoretical properties.
\citet{ionides15} generalized \citet{lindstrom12}'s approach and
combined the idea with data cloning \citep{lele07}, developing a Bayes
map iterated filtering with an entirely different theoretical approach.
\citet{nguyenis215} revisited the approach of \citet{doucet2013derivative},
using a different perturbation noise and computed both the gradient
and the Hessian. Similar to intractable likelihood in the context
of iterated filtering, \citet{Poyiadjis-etal:2009, nemeth2014sequential, doucet2013derivative}
showed that the gradient and Hessian information can also be computed from
particle filter. In the same line, manifold Langevin Monte Carlo (mMALA)
\cite{girolami2011riemann} exploits the Hessian information to simplify
the tedious tuning method while improving on convergence rate. However, this 
relies on rather strong assumptions that the gradient, and Hessian
information of transition density and observation density can be sampled
from. This is quite unrealistic in many real world applications. We,
therefore, followed the formal approaches, based solely on very weak
assumptions of being able to sample from transition density and evaluate
from observation density. Motivated from the fact that the gradient and
Hessian information can be approximated using the first and the second moments
\cite{ionides11, doucet2013derivative}, we propose to use such approximations
in the context of accelerate iterated filtering. Ionides uses score
vector merely while Doucet includes the Hessian information for
the independent white noise, which is not quite useful in the context
of iterated filtering with natural random walk noise. \citet{nguyenis215}
proposed to approximate the gradient and Hessian using random walk noise
to efficiently explore the mode of the likelihood. Other than exploiting
approximations of the Hessian under weak assumption, we chose an alternative
approach. That is, we apply the accelerate gradient approach to the
approximation of the gradient of the log likelihood for an effective
estimation approach.

The key contributions of this paper are three folds. Firstly, we developed
and showed that accelerate iterated filtering algorithm converges using
a general non-increasing step size with bias approximation of the
gradient. It is simple, elegant, and generalizable to faster algorithms.
Secondly, we proved that it has a higher convergence rate in general
convex and non-convex conditions of the objective log likelihood. Finally,
we showed substantial improvements of the method on a toy problem and
on a real world challenge problem of vivax malaria model compared
to previous simulation-based inference approaches.

The paper is organized as follows. In the next section we introduce
some notations and we develop the framework of accelerate iterated
filtering. In Sections \ref{sec:AIF}, we state the convergence of
this approximation method to the true maximum likelihood estimation
by iterating and accelerating noisy gradient of the log likelihood.
We validate the proposed methodology by a toy example and a challenging
inference problem of fitting a malaria transmission model to time
series data in Section \ref{sec:5Experiments}, showing substantial
gains for our methods over current alternatives. We conclude in Section
\ref{sec:5Conclusion} with the suggesting of the future works to
be extended. The proofs are postponed to the Appendix.

%%%%%%%%%%%%%%%%%%%%%%%%%%%%%%%%%%%%%%%%%%%%%%%%%%%%%%%%%%%%%%%%%%
%%%      Section: Notation and problem def       %%%%%%%%%%%%%%%%%
%%%%%%%%%%%%%%%%%%%%%%%%%%%%%%%%%%%%%%%%%%%%%%%%%%%%%%%%%%%%%%%%%%
% \section{Proposed approach}

\section{Background of simulation-based inferences}

We are interested in a general latent variable model since this is an ubiquitous
model for applied sciences. Let $\mathcal{X}$ be a latent state space
with a density $q_{\theta}(x)$ parameterized by $\theta\in\Theta=\mathbb{R}^{d}$,
and let $\mathcal{Y}$ be an observation space equipped with a conditional
density $f_{\theta}(y|x)$. The observation $y\in\mathcal{Y}$ are
considered as fixed and we write the log-likelihood function of
the data $\ell(\theta)\overset{\triangle}{=}\log\int q_{\theta}(x)f_{\theta}(y|x)dx$.
We work with the maximum likelihood estimator, $\hat{\theta}=\arg\max\ell(\theta)$
where $\ell(\theta)$ is intractable but $f_{\theta}(y|x)$ can be
evaluated, by using samples where $f_{\theta}(y|x)$
is also intractable. This process often uses the first order stochastic
approximation \cite{kushner78}, which involves a Monte Carlo approximation
to a difference equation, $\theta_{m}=\theta_{m-1}+\gamma_{m}\nabla\ell(\theta_{m-1}),$
where $\theta_{0}\in\Theta$ is an arbitrary initial estimate and
$\{\gamma_{m}\}_{m\geq1}$ is a sequence of step sizes with ${\sum_{m\geq1}\gamma_{m}=\infty}$
and ${\sum_{m\geq1}\gamma_{m}^{2}<\infty}$. The algorithm converges
to a local maximum of $\ell(\theta)$ under regularity conditions.
The term $\nabla\ell(\theta)$, also called the score function, is shorthand
for the $\mathbb{R}^{d}$-valued vector of partial derivatives, $\nabla\ell(\theta)=\frac{\partial\ell(\theta)}{\partial\theta}$. 

Sequential Monte Carlo (SMC) approaches have previously been developed
to estimate the score function \citep{Poyiadjis-etal:2009,nemeth2013particle,DahlinLindstenSchon2015a}.
However, under the simulation-based setting, which does not require the
ability to evaluate transition densities and their derivatives, these
approaches are not applicable. As a result, \cite{ionides11}, \cite{doucet2013derivative}
used an artificial dynamics approach to estimate the derivatives. Specifically, 
\cite{nguyenis215} considers a parametric model consisting of a
density $p_{Y}(y;\theta)$ with the log-likelihood of the data $y^{*}\in\mathcal{Y}$
given by $\ell(\theta)=\log p_{Y}(y^{*};\theta)$. A stochastically
perturbed model corresponding to a pair of random variables $(\breve{\Theta},\breve{Y})$
having a joint probability density on $\mathbb{R}^{d}\times\mathcal{Y}$
can be defined as $p_{\breve{\Theta},\breve{Y}}(\breve{\vartheta},\ y;\theta,\ \tau)=\tau^{-d}\kappa\left\{ \tau^{-1}(\breve{\vartheta}-\theta)\right\} p_{Y}(y;\breve{\vartheta}).$ 
Suppose the following regularity conditions, identical to the assumptions
of \cite{doucet2013derivative}:
\begin{assumption}\label{ass1} There exists $C<\infty$ such that
for any integer $k\geq1,1\leq i_{1},\ \ldots,\ i_{k}\leq d$ and $\beta_{1},\ \ldots,\ \beta_{k}\geq1$,
$\int\left|u_{i_{1}}^{\beta_{1}}u_{i_{2}}^{\beta_{2}}\cdots u_{i_{k}}^{\beta_{k}}\right|\kappa(u)\ du\leq C,$
where $\kappa$ is a symmetric probability density on $\mathbb{R}^{d}$
with respect to Lebesgue measure and $\Sigma=(\sigma_{i,j})_{i,j=1}^{d}$
is the non-singular covariance matrix associated to $\kappa$. \end{assumption}
\begin{assumption}\label{ass2} There exist $\gamma,\ \delta,\ M>0,$
such that for all $u\in\mathbb{R}^{d}$, $|u|>M\Rightarrow\kappa(u)<e^{-\gamma|u|^{\delta}}.$
\end{assumption}
\begin{assumption}\label{ass3} $\ell$ is four times continuously
differentiable and $\delta$ defined as in Assumption \ref{ass2}.
For all $\theta\in\mathbb{R}^{d}$, there exists $0<\eta<\delta,\ \epsilon,\ D>0,$
such that for all $u\in\mathbb{R}^{d}$, $\mathcal{L}(\theta+u)\leq De^{\epsilon|u|^{\eta}},$
where $\mathcal{L}$ : $\mathbb{R}^{d}\rightarrow\mathbb{R}$ is the
associated likelihood function $\mathcal{L}=\exp\ell$. \end{assumption}
Under these regularity assumptions, \cite{doucet2013derivative}
show that 
\begin{equation}
\left|\tau^{-2}\Sigma^{-1}{\mathbb{E}}\left(\breve{\Theta}-\theta\left|\breve{Y}=y^{*}\right.\right)-\nabla\ell\left(\theta\right)\right|<C\tau^{2}.\label{eq:4.1}
\end{equation}
These approximations are useful for latent variable models, where
the log-likelihood of the model consists of marginalizing over a latent
variable, $X$, 
$$\ell(\theta)=\log\int{p_{{X},{Y}}(x,y^{*};\theta)\, dx}.$$
In this case, the expectations in equation \ref{eq:4.1} can be approximated
by Monte Carlo importance sampling, as proposed by \cite{ionides11}
and \cite{doucet2013derivative}. 
In \cite{nguyenis215}, the POMP model is a specific latent variable
model with ${X}=X_{0:N}$ and ${Y}=Y_{1:N}$. A perturbed POMP model
is defined to have a similar construction to our perturbed latent
variable model with ${\breve{X}}=\breve{X}_{0:N}$, $\breve{Y}=\breve{Y}_{1:N}$
and $\breve{\Theta}=\breve{\Theta}_{0:N}$. \citep{ionides11} perturbed
the parameters by setting $\breve{\Theta}_{0:N}$ to be a random walk
starting at $\theta$, whereas \citep{doucet2013derivative} took
$\breve{\Theta}_{0:N}$ to be independent additive white noise perturbations
of $\theta$. We take advantage of the asymptotic developments of
\citep{doucet2013derivative} while maintaining some practical advantages
of random walk perturbations for finite computations, so we use the
construct $\breve{\Theta}_{0:N}$ as in \cite{nguyenis215} as follows.

Let $Z_{0},\ \ldots,\ Z_{N}$ be $N+1$ independent draws from a density
$\psi$. \cite{nguyenis215} introduces $N+2$ perturbation parameters,
$\tau$ and $\tau_{0},\ldots,\tau_{N}$, and construct a process $\breve{\Theta}_{0:N}$
by setting ${\breve{\Theta}_{n}=\theta+\tau\sum_{i=0}^{n}\tau_{i}Z_{i}}$ for $0\le n\le N$. 
We later consider a limit where $\tau_{0:N}$
as fixed and the scale factor $\tau$ decreases toward zero, and subsequently
another limit where $\tau_{0}$ is fixed but $\tau_{1:N}$ decrease
toward zero together with $\tau$. Let $p_{\breve{\Theta}_{0:N}}(\breve{\vartheta}_{0:N};\theta,\ \tau,\ \tau_{0:N})$
be the probability density of $\breve{\Theta}_{0:N}$. We define
the artificial random variables $\breve{\Theta}_{0:N}$ via their
density, 
\begin{multline*}
p_{\breve{\Theta}_{0:N}}(\breve{\vartheta}_{0:N};\theta,\ \tau,\ \tau_{0:N})=\\
(\tau\tau_{0})^{-d}\psi\left\{ (\tau\tau_{0})^{-1}(\breve{\vartheta_{0}}-\theta)\right\}\times\prod_{n=1}^{N}(\tau\tau_{n})^{-d}\psi\left\{ (\tau\tau_{n})^{-1}(\breve{\vartheta}_{t}-\breve{\vartheta}_{t-1})\right\}.
\end{multline*}
We define the stochastically perturbed model with a Markov process
$\{(\breve{X}_{n},\breve{\Theta}_{n}),\ 0\leq n\leq N\}$, observation
process $\breve{Y}_{1:N}$ and parameter $(\theta,\ \tau,\ \tau_{0:N})$
by the factorization of their joint probability density 
\begin{multline*}
p_{\breve{X}_{0:N},\breve{Y}_{1:N},\breve{\Theta}_{0:N}}(x_{0:N},y_{1:N},\breve{\vartheta}_{0:N};\theta,\ \tau,\ \tau_{0:N})\\
=p_{\breve{\Theta}_{0:N}}(\breve{\vartheta}_{0:N};\theta,\ \tau,\ \tau_{0:N})p_{\breve{X}_{0:N},\breve{Y}_{1:N}|\breve{\Theta}_{0:N}}(x_{0:N},\ y_{1:N}|\breve{\vartheta}_{0:N}),
\end{multline*}
where 
\begin{multline*}
p_{\breve{X}_{0:N},\breve{Y}_{1:N}|\breve{\Theta}_{0:N}}(x_{0:N},y_{1:N}|\breve{\vartheta}_{0:N};\theta,\tau,\ \tau_{0:N})=\\
\mu(x_{0};\breve{\vartheta}_{0})\prod_{n=1}^{N}f_{n}(x_{n}|x_{n-1};\breve{\vartheta}_{n})\prod_{n=1}^{N}g_{n}(y_{n}|x_{n};\breve{\vartheta}_{n}). 
\end{multline*}
This extended model can be used to define a perturbed parameter log-likelihood
function, defined as 
\begin{equation}
\breve{\ell}(\breve{\vartheta}_{0:N})=\log p_{\breve{Y}_{1:N}|\breve{\Theta}_{0:N}}(y_{1:N}^{*}|\breve{\vartheta}_{0:N};\theta,\tau,\tau_{0:N}).\label{eq:eloglik-1}
\end{equation}
Here, the right hand side does
not depend on $\theta$, $\tau$ or $\tau_{0:N}$. We have designed
(\ref{eq:eloglik-1}) so that, setting $\breve{\vartheta}^{[N+1]}=(\theta,\theta,\dots,\theta)\in\mathbb{R}^{d(N+1)},$ the log-likelihood of the unperturbed model can be written as 
$\ell(\theta)=\breve{\ell}(\breve{\vartheta}^{[N+1]}).$
For the perturbed likelihood, we need an additional assumption of
the extended version.
\begin{assumption}\label{ass5-1} $\breve{\ell}$ is four times continuously
differentiable. For all $\theta\in\mathbb{R}^{d}$, there exist $\epsilon>0$,
$D>0$ and $\delta$ defined as in Assumption \ref{ass2}, such that
for all $0<\eta<\delta$ and $u_{0:N}\in\mathbb{R}^{d(N+1)}$,$\breve{\mathcal{L}}(\breve{\vartheta}^{[N+1]}+u_{0:N})\leq De^{\epsilon\sum_{n=1}^{N}|u_{n}|^{\eta}},$
where $\breve{\mathcal{L}}(\breve{\vartheta}_{0:N})=\exp\{\breve{\ell}(\breve{\vartheta}_{0:N})\}$
is the perturbed likelihood. \end{assumption} Let $\breve{\mathbb{E}}_{\theta,\tau,\tau_{0:N}}$,
${\mathrm{\breve{C}ov}}_{\theta,\tau,\tau_{0:N}}$, ${\mathrm{\breve{V}ar}}_{\theta,\tau,\tau_{0:N}}$
denote the expectation, covariance and variance with respect to the
associated posterior, $p_{\breve{\Theta}_{0:N}|\breve{Y}_{1:N}}(\breve{\vartheta}_{0:N}|y_{1:N}^{*};\theta,\ \tau,\tau_{0:N}).$
By using $\breve{\mathbb{E}}$, ${\mathrm{\breve{C}ov}}$, ${\mathrm{\breve{V}ar}}$
instead of $\breve{\mathbb{E}}_{\theta,\tau,\tau_{0:N}}$, ${\mathrm{\breve{C}ov}}_{\theta,\tau,\tau_{0:N}}$,
${\mathrm{\breve{V}ar}}_{\theta,\tau,\tau_{0:N}}$ respectively, a
theorem similar to theorem 4 of \cite{doucet2013derivative} but
for random walk noise instead of independent white noise is derived.
\begin{theorem}  \label{thm1-1} [Theorem 2 of \cite{nguyenis215}] Suppose Assumptions \ref{ass1}, \ref{ass2}
and \ref{ass5-1}, there exists a constant $C$ independent of $\tau,\tau_{1},...\tau_{N}$
such that, 
\[
\left|\nabla\ell\left(\theta\right)-\tau^{-2}\Psi^{-1}\left\{ \tau_{0}^{-2}\breve{\mathbb{E}}\left(\breve{\Theta}_{0}-\theta|\breve{Y}_{1:N}=y_{1:N}^{*}\right)\right\} \right|<C\tau^{2},
\]
where $\Psi$ is the non-singular covariance matrix associated to
$\psi$. \end{theorem}

Theorem \ref{thm1-1} formally allows an approximation of $\nabla{\ell}\left(\theta\right)$.
\cite{nguyenis215} also presents an alternative variations on these
results which lead to more stable Monte Carlo estimation. \begin{theorem}
\label{thm3}[Theorem 3 of \cite{nguyenis215}] Suppose Assumption
\ref{ass1}, \ref{ass2} and \ref{ass5-1} hold. In addition, assume
that $\tau_{n}=O(\tau^{2})$ for all $n=1\ldots N$, the following
holds true, 
\begin{equation}
\left|\nabla\ell\left(\theta\right)-\frac{1}{N+1}\tau^{-2}\tau_{0}^{-2}\Psi^{-1}\sum_{n=0}^{N}\left\{ \breve{\mathbb{E}}\left(\breve{\Theta}_{n}-\theta|\breve{Y}_{1:N}=y_{1:N}^{*}\right)\right\} \right|=O(\tau^{2}).
\end{equation}
\end{theorem}

These theorems are useful for our approaches because we can approximate
the gradient of the log-likelihood of the extended model to the second
order of $\tau$ which we will later show that it fits well with our
accelerate simulation based setup.

\section{Proposed accelerate iterated filtering} \label{sec:AIF} 
Our motivation comes from the accelerated gradient method for smooth non-linear
stochastic programming literature. By using an approximation of the
score function, it is possible to use an accelerated gradient method
as in Nesterov acceleration scheme in optimization literature. One issue
with the accelerated gradient approach is that it is not clear how
the technique can be used in situations where both the likelihood
and the gradient are intractable. These sorts of examples are common
in scientific applications of state space models where the state process
is a diffusion process or an ordinary differential equation (ODE)
with stochastic coefficients. However, in these family
of iterated filtering approaches, the score function can be approximated
with noise under control without affecting the convergence rate. Specifically,
applying an accelerated inexact gradient algorithm in the iterated
filtering approach can obtain an optimal rate of convergence.

In this paper, $\epsilon_{k}$ denotes the error in the approximation
of the gradient. Using the same notation as \cite{ghadimi2016accelerated},
denote the sequences of magnitudes of the errors in the gradient approximations
$\left\{ \left\Vert \epsilon_{k}\right\Vert \right\} $. Suppose the
following assumptions: \begin{assumption}\label{ass6} The function
$\ell$ : $\Theta\rightarrow\mathbb{R}$ is differentiable, bounded
from above and has a L-Lipschitz-continuous gradient, i.e. for all
$\theta,\ \vartheta\in\Theta$, $\left\Vert \nabla\ell(\theta)-\nabla\ell(\vartheta)\right\Vert \leq L\left\Vert \theta-\vartheta\right\Vert ,$
where $\nabla\ell$ denotes the gradient of $\ell$. The function
$\ell$ attains its maximum at a certain $\theta^{*}\in\Theta$. \end{assumption}
In the sequel, $\Theta$ denotes a finite-dimensional Euclidean space
with norm $\left\Vert \cdot\right\Vert $ and inner product $\left\langle \cdot,\cdot\right\rangle $.
It can be shown that (e.g. in \cite{nesterov2005smooth}) Assumption \ref{ass6}
is equivalent to 
\begin{equation}
\left|\ell(\vartheta)-\ell(\theta)-\left\langle \nabla\ell(\theta),\mathrm{\vartheta}-\theta\right\rangle \right|\leq\frac{L}{2}\left\Vert \vartheta-\theta\right\Vert ^{2},\;{\displaystyle \forall\theta,\ \vartheta\in\Theta}\label{eq:2.1}
\end{equation}
It is well-known that the gradient descent method converges for a
general non-convex optimization problem but it does not achieve the
optimal rate of convergence, in terms of the functional optimality
gap, when $\ell(\cdot)$ is convex \cite{ghadimi2016accelerated}.
In contrast, the accelerated gradient method in \cite{nesterov2013introductory}
is optimal for solving convex optimization problems, but does not
necessarily converge for solving nonconvex optimization problems.
\cite{ghadimi2016accelerated} proposed a modified accelerated gradient
method which can converge in both convex and non-convex optimization
problem. However, they assumed unbiased estimation of the gradient
which is not satisfied for most simulation-based inferences. Below,
we extend the approach of Ghadimi to an accelerated inexact gradient
(AIG) method in the context of accelerate iterated filtering. That is, we allow bias in gradient approximation
by properly specifying the stepsize policy. We prove that it not only
achieves the same optimal rate of convergence for both convex and non-convex optimizations,
but also exhibits the best-known rate of convergence for simulation-based
inference problems.
\begin{algorithm}[H]
\caption{Accelerate Inexact Gradient (AIG)} 
\label{alg0}
\begin{algorithmic}[1]
\Statex
\INPUT  
\Statex $\theta_{0}\in\Theta.$  
\Statex $\left\{ \beta_{\mathrm{k}}>0\right\}$, $\left\{ \lambda_{k}>0\right\}$  for any $k\geq2$.
\Statex $\left\{ \alpha_{k}\right\} \in\left(0,1\right)$ for $k>1$ and $\alpha_{1}=1$.
\newline
\State    $\theta_{0}^{ag}=\theta_{0}$.                   \Comment Initialize 
\For {$k$ in $1...N$}
    \State \begin{equation}\theta_{k}^{md}=(1-\alpha_{k})\theta_{k-1}^{ag}+\alpha_{k}\theta_{k-1}\label{eq:2.2}\end{equation} 
    \State  \begin{equation}\theta_{k}=\theta_{k-1}-\lambda_{k}\left(\widehat{\nabla\ell(\theta_{k}^{md})}\right)\label{eq:2.3}\end{equation} 
    \State  \begin{equation}\theta_{k}^{ag}=\theta_{k-1}^{md}-\beta_{k}\left(\widehat{\nabla\ell(\theta_{k}^{md})}\right)\label{eq:2.4}\end{equation} \Comment where $\widehat{\nabla\ell(\theta_{k}^{md})}$ is an estimation of
$\nabla\ell(\theta_{k}^{md})$ with error $\epsilon_{k}$.
\EndFor
\end{algorithmic}
\end{algorithm}
In addition to Assumption \ref{ass6}, we assume a noise control condition
for Algorithm \ref{alg0}. 
\begin{assumption}\label{ass7}
$\Theta$ is bounded. There exists an $A<\infty$ such that $\sum_{k=1}^{N}\lambda_{k}\left\Vert \epsilon_{k}\right\Vert <A.$ 
\end{assumption}
Given some mild conditions often satisfied by controlling the artificial
noises, we have the following result.
\begin{theorem}\label{thm4}
(Extension of Theorem 1 of \cite{ghadimi2016accelerated}).\\
Suppose Assumptions \ref{ass6} and \ref{ass7} hold. In addition, let $\{\theta_{k},\ \theta_{k}^{ag}\}$ $k\geq1$ be computed by Algorithm \ref{alg0}.\\
a) If sequences $\left\{ \alpha_{k}\right\} ,\left\{ \beta_{k}\right\} $,
$\left\{ \lambda_{k}\right\} $ and $\left\{ \Gamma_{k}\right\} $
satisfy
\begin{equation}
\Gamma_{k}:=\begin{cases}
1 & k=1\\
(1-\alpha_{k})\Gamma_{k-1} & k\geq2
\end{cases},\ \label{eq:2.6}
\end{equation}
\begin{equation}
C_{k}:=1-L{\displaystyle \lambda_{k}-\frac{L(\lambda_{k}-\beta_{k})^{2}}{2\lambda_{k}\alpha_{k}\Gamma_{k}}\left(\sum_{\tau=k}^{N}\frac{1}{\Gamma_{\tau}}\right)>0, \mbox{ for}\ 1\leq k\leq N}\label{eq:2.7},
\end{equation}
then for any $N\geq1$, we have for some $B<\infty$,
\begin{equation}
{\displaystyle \min_{k=1,...,N}\left\Vert \nabla\ell(\theta_{k}^{md})+\epsilon_{k}\right\Vert ^{2}\leq\frac{\ell(\theta_{0})-\ell^{*}+B}{\sum_{k=1}^{N}\lambda_{k}C_{k}}}.\label{eq:2.8}
\end{equation}
b) Suppose that $\ell(\cdot)$ is convex. If sequences $\left\{ \alpha_{k}\right\} ,\left\{ \beta_{k}\right\} $,$\left\{ \lambda_{k}\right\} $
and $\left\{ \Gamma_{k}\right\} $ satisfy
\begin{equation}
\alpha_{k}\lambda_{k}\leq\beta_{k}<\frac{1}{L}\label{eq:2.9},
\end{equation}
\begin{equation}
\frac{\alpha_{1}}{\lambda_{1}\Gamma_{1}}\geq\frac{\alpha_{2}}{\lambda_{2}\Gamma_{2}}\geq\ldots\label{eq:2.10},
\end{equation}
then for any $N\geq1$, we have
\begin{multline}
 \min_{k=1,...,N}\left\Vert \nabla\ell(\theta_{k}^{md})+\epsilon_{k}\right\Vert ^{2}\\
\leq2\frac{\frac{\Vert \theta^{*}-\theta_{0}\Vert^{2}}{2\lambda_{1}}+\sum_{k=1}^{N}\Gamma_{k}^{-1}\left[\beta_{k}\left\Vert \epsilon_{k}\right\Vert \left\Vert \nabla\ell(\theta_{k}^{md})+\epsilon_{k}\right\Vert +\alpha_{k}\left\Vert \epsilon_{k}\right\Vert \left\Vert \theta_{k-1}-\theta_0\right\Vert \right]}{\sum_{k=1}^{N}\Gamma_{k}^{-1}\beta_{k}(1-L\beta_{k})},\label{eq:2.11}
\end{multline}
\begin{multline}
\ell(\theta_{N}^{ag})-\ell(\theta^{*})\\
\leq\Gamma_{N}\left[\frac{\left\Vert \theta_{0}-\theta^{*}\right\Vert ^{2}}{\lambda_{1}}+\sum_{k=1}^{N}\Gamma_{k}^{-1}\left[\beta_{k}\left\Vert \epsilon_{k}\right\Vert \left\Vert \nabla\ell(\theta_{k}^{md})+\epsilon_{k}\right\Vert +\alpha_{k}\left\Vert \epsilon_{k}\right\Vert \left\Vert \theta_{k-1}-\theta_0\right\Vert \right]\right]\label{eq:2.12}.
\end{multline}
\end{theorem}
There are various options for selecting $\left\{ \alpha_{k}\right\} ,\left\{ \beta_{k}\right\} $,$\left\{ \lambda_{k}\right\} $,$\left\{ \Gamma_{k}\right\}$.  
By controlling error $\epsilon_{k}$, we can provide some of these
selections below which guarantee the optimal convergence rate of the
AIG algorithm for both convex and nonconvex problems.
\begin{theorem}\label{thm5}
Suppose Assumptions \ref{ass6} and \ref{ass7} hold. In addition, suppose that $\left\{ \beta_{k}\right\} $ in the accelerated
gradient method are set to $\beta_{k}=\frac{1}{2L}$.

a) If sequences $\left\{ \alpha_{k}\right\}$ and $\left\{ \lambda_{k}\right\}$
satisfy
\begin{equation}
\lambda_{k}\in\left[\beta_{k},(1+\frac{1}{k})\beta_{k}\right],  \mbox{ for }{\displaystyle \;\forall k\geq1}\label{eq:2.28},
\end{equation}
then for any $N\geq1$, we have
\begin{equation}
{\displaystyle \min_{k=1,\ldots, N}\Vert\nabla\ell(\theta_{k}^{md})+\epsilon_{k}\Vert^{2}\leq O\left(\frac{1}{N}\right)}\label{eq:2.29}.
\end{equation}
Suppose that $\epsilon_{k}=O\left(\tau^{2}\right)\leq O(\frac{1}{k})$,
then the AIG method can find a solution $\bar{\theta}$ such
that $\left\Vert \nabla\ell(\bar{\theta})\right\Vert ^{2}\leq\epsilon$
in at most $O(1/\epsilon^{2})$ iterations.

b) Suppose that $\ell(\cdot)$ is convex and $\epsilon_{k}=O\left(\tau^{2}\right)\leq O(\frac{1}{k^{2+\delta+\delta_{1}}})$
for some $\delta_{1}>0$. If $\left\{ \lambda_{k}\right\} $ satisfies
\begin{equation}
{\displaystyle \lambda_{k}=\left(k^{1+\delta}-\left(k-1\right)^{1+\delta}\right)\:\forall k\geq1},\label{eq:2.30}
\end{equation}
then for any $N\geq1$, we have
\begin{equation}
\min_{k=1,...,N}\left\Vert \nabla\ell(\theta_{k}^{md})+\epsilon_{k}\right\Vert ^{2}{\displaystyle \leq}O\left(\frac{1}{N^{2+\delta}}\right),\label{eq:2.31}
\end{equation}
\begin{equation}
{\displaystyle \ell(\theta_{N}^{ag})-\ell(\theta^{*})\leq O\left(\frac{1}{N^{1+\delta}}\right)},\label{eq:2.32}
\end{equation}
then the AIG method can find a solution $\bar{\theta}$ such
that $\left\Vert \nabla\ell(\bar{\theta})\right\Vert ^{2}\leq\epsilon$
in $O\left(1/\epsilon^{\frac{1}{2+\delta}}\right)$ at most.
\end{theorem}
\begin{algorithm}[H]
\caption{Accelerate Iterated Filtering (AIF)} 
\label{alg1}
\begin{algorithmic}[1]
\Statex
\INPUT  
\Statex Starting parameter, $\theta_0=\theta^{ag}_0$ , sequences, $\alpha_n, \beta_n, \lambda_n, \Gamma_n$ 
\Statex simulator for $f_{X_0}(x_0|\theta)$, $f_{X_n|X_{n-1}}(x_n| x_{n-1}|\theta)$, evaluator for $f_{Y_n|X_n}(y_n| x_{n}|\theta)$
\Statex data, $y^*_{1:N}$, labels designating IVPs, $I\subset\{1,\dots,p\}$, initial scale multiplier, $C>0$
\Statex number of particles, $J$, number of iterations, $M$, cooling rate, $0<a<1$, perturbation scales, $\sigma_{1:p}$
\Ensure
\Statex Maximum likelihood estimate $\theta_{MLE}$
\newline
\State    $\theta^{md}_0=\theta_0$                   \Comment Initialize \label{alg1:init:perturb}
\State    $[\Theta^F_{0,j}]_i \sim N\left([\theta^{md}_0]_i, (C a^{m-1} \sigma_i)^2\right)$ for $i$ in $1..p$, $j$ in $1... J$. \Comment Initialize filter mean for parameters
\State    simulate $X_{0,j}^F \sim f_{X_0}\big(\cdot;{\Theta^F_{0,j}}\big)$ for $j$ \textrm{in} $1..J$. \Comment Initialize states \label{alg1:initstates}
\For {$m$ in $1...M$}
    \State    $\theta^{md}_m= (1-\alpha_m)\theta^{ag}_{m-1}+\alpha_{m_1}\theta_{m-1}$.
    \For {$n$ \textrm{in} $1... N$}
	\State      $\big[\Theta_{n,j}^{P}\big]_{i}\sim\mathcal{N}\big(\big[\Theta^F_{n-1,j}\big]_{i},(c^{m-1}\sigma_{i})^{2}\big)$ for $i\notin I$, $j$ in $1:J$.\;\Comment Perturb \label{alg1:perturb}
	\State      ${X}_{n,j}^{P}\sim{f}_n\big({x}_{n}|{X}_{n-1,j}^{F};{\Theta_{n,j}^{P}}\big)$ for $j$ in $1:J$.\; \Comment Simulate prediction particles \label{alg1:sim}
	\State      $w(n,j)=g_n(y_{n}^{*}|X_{n,j}^{P};\Theta_{n,j}^{P})$ for $j$ in $1:J$.\; \Comment Evaluate weights \label{alg1:weights}
	\State      $\breve{w}(n,j)=w(n,j)/\sum_{u=1}^{J}w(n,u)$.\; \Comment Normalize weights \label{alg1:normalize}
	\State      $k_{1:J}$ with $P\left\{ k_{u}=j\right\} =\breve{w}\left(n,j\right)$.\label{alg1:syst}\; \Comment Apply systematic resampling to select indices
	\State      $X_{n,j}^{F}=X_{n,k_{j}}^{P}$ and $\Theta_{n,j}^{F}=\Theta_{n,k_{j}}^{P}$ for $j$ in $1:J$.\;  \Comment Resample particles \label{alg1:resample}
    \EndFor
\State    $S_{m}=c^{-2(m-1)}\Psi^{-1}\sum_{n=1}^N\big[\left(\bar{\theta}_{n}-\theta^{md}_{m-1}\right)\big]$ \Comment Update Parameters \label{alg1:updateivps}
\State    $\big[\theta_{m}\big]_i=\theta_{m-1}-\lambda_{m-1}\big[S_{m}\big]_{i}$ for $i \notin I$.\;
\State    $\big[\theta^{ag}_{m}\big]_i=\theta^{md}_{m-1}-\beta_{m-1}\big[S_{m}\big]_{i}$ for $i \notin I$.\;
\State    $\big[\theta_{m}\big]_i=\frac{1}{J}\sum_{j=1}^J\big[\Theta^F_{L,j}\big]_i$ for $i \in I$.\;

\EndFor

\end{algorithmic}
\end{algorithm}
We now add a few remarks about the extension results obtained in Theorem \ref{thm5}. First, 
if the problem is convex, by choosing more aggressive stepsizes
$\{\lambda_{k}\}$ in \eqref{eq:2.30}, the AIG method exhibits the optimal
rate of convergence in \eqref{eq:2.32}. It is also worth noting that
with such a selection of $\{\lambda_{k}\}$, the AIG method can find
a solution $\bar{\theta}$ such that $\left\Vert \nabla\ell(\bar{\theta})\right\Vert ^{2}\leq\epsilon$
in at most $O(1/\epsilon^{1/2+\delta})$ iterations. The latter result has been shown by \cite{nesterov2005smooth},
\cite{ghadimi2016accelerated} but only for the accelerate unbiased gradient method. Second, 
observe that $\left\{ \lambda_{k}\right\} $ in \eqref{eq:2.28} for
general nonconvex problems is in the order of $O(1/L)$,
while the one in \eqref{eq:2.30} for convex problems are more aggressive
(in $O(k/L))$. The value $\delta$ is optimal at $1$ for
convergence rate. However, it may not be optimal
for computation of controlling the noises.  Finally, we show that we can apply the stepsize policy in \eqref{eq:2.28}
for solving general inexact gradient problems for both convex and nonconvex optimization. 
The sequential Monte Carlo filter can be arbitrarily approximated to the exact filter 
by choosing sufficiently large number of particles \citep{ionides11}.
It can be seen that we can choose the perturbation sequence so that the gradient noise satisfies
condition in Theorem~4. For completeness, we present the pseudo code of the proposed algorithm as in Algorithm \ref{alg1}.

%%%%%%%%%%%%%%%%%%%%%%%%%%%%%%%%%%%%%%%%%%%%%%%%%%%%%%%%%%%%%%%%%%
%%%      Experiments			         %%%%%%%%%%%%%%%%%
%%%%%%%%%%%%%%%%%%%%%%%%%%%%%%%%%%%%%%%%%%%%%%%%%%%%%%%%%%%%%%%%%%
\section{Numerical examples\label{sec:5Experiments}}
To measure the performance of the new inference algorithm, we evaluate
our accelerate iterated filtering on some benchmark examples and compare
it to the existing simulation-based approaches. We make use
of well tested and maintained code of R \cite{amanual} packages such
as pomp \cite{king15pomp}. Specifically, models are coded using C
snipet declarations \cite{king15pomp}. New algorithm is written in
R package is2, which provides user friendly interfaces in R and efficient
matrix operations in the highly optimized Rcpp \cite{eddelbuettel2011rcpp}.
All the simulation-based approaches mentioned above use sequential
Monte Carlo algorithm (SMC), implemented using bootstrap filter.
Experiments were carried out on a cluster of $32$ cores Intel
Xeon E5-2680 $2.7$ Ghz with $256$ GB memory. For a fair comparison,
we try to use the same setup  and assessment for every inference method. 
A public Github repository containing scripts for reproducing our
results may be found at https://github.com/nxdao2000/AIFcomparisons.

\subsection{Toy example: A linear, Gaussian model}

In this subsection, we compare our accelerate iterated filtering algorithm
to the original iterated filtering algorithm IF1 \cite{ionides06-pnas},
Bayes map iterated filtering (IF2) \cite{ionides15} and the second-order
iterated smoothing (IS2) \cite{nguyenis215}. It has been shown in \cite{nguyenis215} and \cite{ionides15} that 
the second-order iterated smoothing with white noise (IS1) \cite{doucet2013derivative} and particle Markov chain Monte Carlo (PMCMC) \cite{andrieu10}
do not perform as well as Bayes map iterated filtering so we leave them out.   
For a computationally convenient setting, simple models provide an opportunity
to test the basic features of inference algorithms. Therefore, we
first consider a bivariate discrete time Gaussian autoregressive process,
a relatively simple mechanistic model. This model is chosen so that
the Monte Carlo calculations can be verified using a Kalman filter.
For this example, there are some alternatives to iterated filtering
class. For example, EM and MCMC algorithms would be practical in this
case although they do not scale well to large dynamic models, so we
do not include them here. The model is given by the state space forms:
$X_{n}|X_{n-1}=x_{n-1}\sim\mathcal{N}(\alpha x_{n-1},\sigma^{\top}\sigma),$
$Y_{n}|X_{n}=x_{n}\sim\mathcal{N}(x_{n},I_{2})$ where $\alpha$,
$\sigma$ are $2\times2$ matrices and $I_{2}$ is $2\times2$ identity
matrix. The data are simulated from the following parameters: 
\[
{\displaystyle \alpha=\left[\begin{array}{cc}
\alpha_{1} & \alpha_{2}\\
\alpha_{3} & \alpha_{4}
\end{array}\right]=\left[\begin{array}{cc}
0.8 & -0.5\\
0.3 & 0.9
\end{array}\right],\ \sigma=\left[\begin{array}{cc}
3 & 0\\
-0.5 & 2
\end{array}\right].}
\]
The number of time points $N$ is set to $100$ and initial starting
point $X_{0}=\left(-3,4\right)$. For each method mentioned above,
we estimate parameters $\alpha_{2}$ and $\alpha_{3}$ for this model
using $J=1000$ particles and run our estimation for $M=25$ iterations.
We start the initial search uniformly on a large rectangular region
$[-1,1]\times[-1,1]$. As can be seen from Fig. 1, all of the distributions
of estimated maximized log likelihoods touch the true MLE (computed
from Kalman filter) at the vertical broken line, implying that they all
successfully converged. The results show that AIF is the most efficient
method of all because using AIF the results have higher mean and smaller
variance compared to other approaches, indicating a higher empirical
convergence rate. Algorithmically, AIF has similar computational costs
with the first order approaches IF1, IF2, and is cheaper than the
second order approach IS2. In deed, average computational
time of twenty independent runs of each approach is given in Table \ref{table:1}.
Additional overheads for estimating score make the computation time
of AIF a bit larger compared to computational time of IF2. However,
with complex models and large enough number of particles, these overheads
become negligible and computational time of AIF will be similar to
other first order approaches. The fact that it has the convergence
rate of second order with computation complexity of first-order shows
that it is a very promising algorithm. In addition, the results also imply that AIF is robust to initial starting
guesses. 

\begin{table}
\begin{center}
\caption{Computation times, in seconds, for the toy example.}
\label{table:1}
\vspace{2mm}
\begin{tabular}{lrrr}
 \hline
 & \rule[-1.5mm]{0mm}{7mm}\hspace{6mm}$J=100$
 & \hspace{6mm}$J=1000$ & \hspace{6mm} $J=10000$ \\ 
  \hline
IF1  & 1.656  & 5.251  & 62.632 \\ 
IF2  & 1.591  & 5.156  & 61.072 \\ 
IS2  & 2.530  & 10.198  & 135.248 \\ 
AIF & 2.729 & 10.278 & 132.016\\
   \hline
\end{tabular}
\end{center}
\end{table}
\begin{figure}
\vspace{-0.5cm}
 \includegraphics[width=7cm,height=7cm]{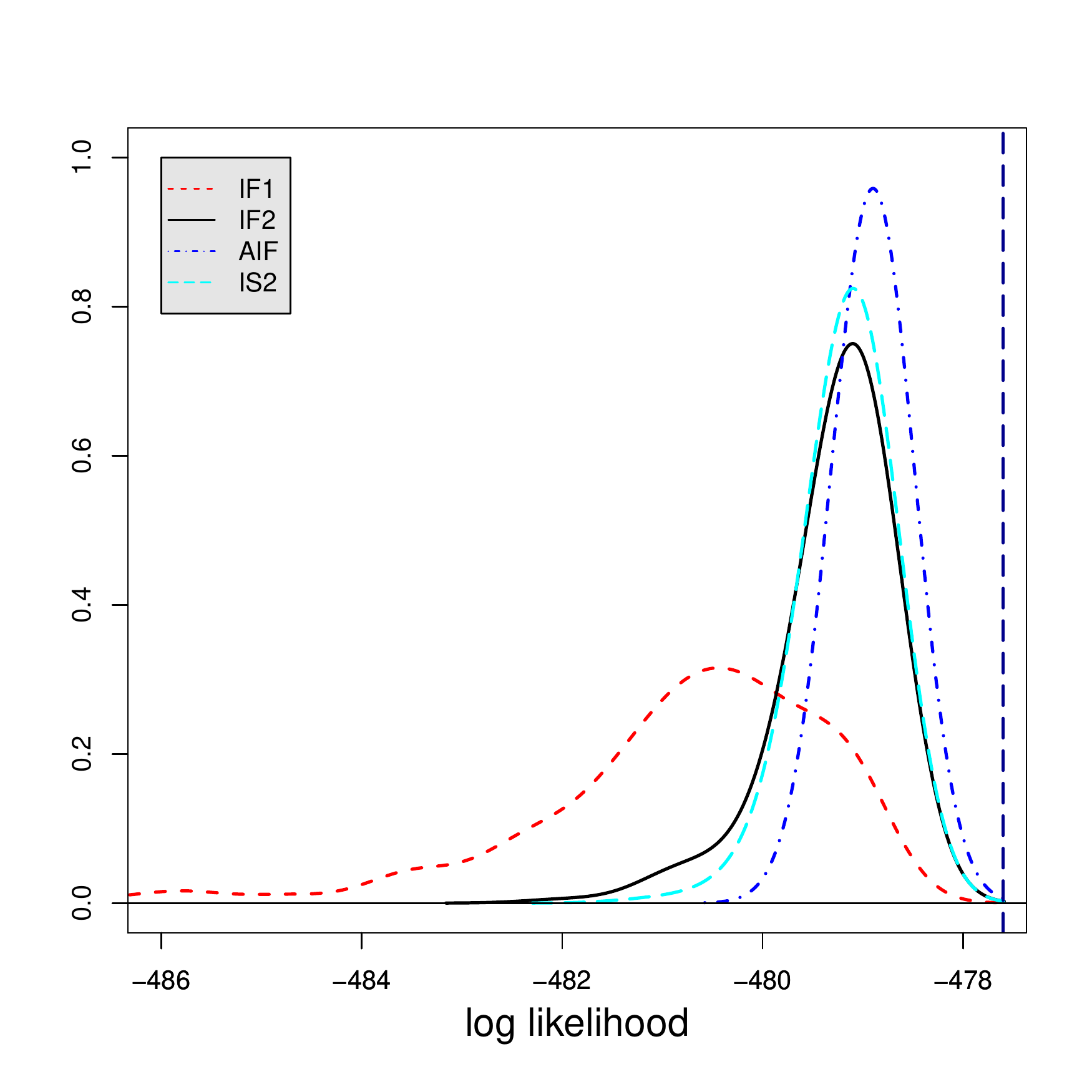} 
\caption{Comparison of estimators for the linear, Gaussian toy example, showing
the densities of the MLEs estimated by the IF1, IF2, AIF and
IS2 methods. The parameters $\alpha_{2}$ and $\alpha_{3}$ were
estimated, started from 200 randomly uniform initial values over a
large rectangular region $[-1,1]\times[-1,1]$. }
\label{fig:toy2} 
\end{figure}
\begin{figure}
\includegraphics[width=12cm,height=12cm]{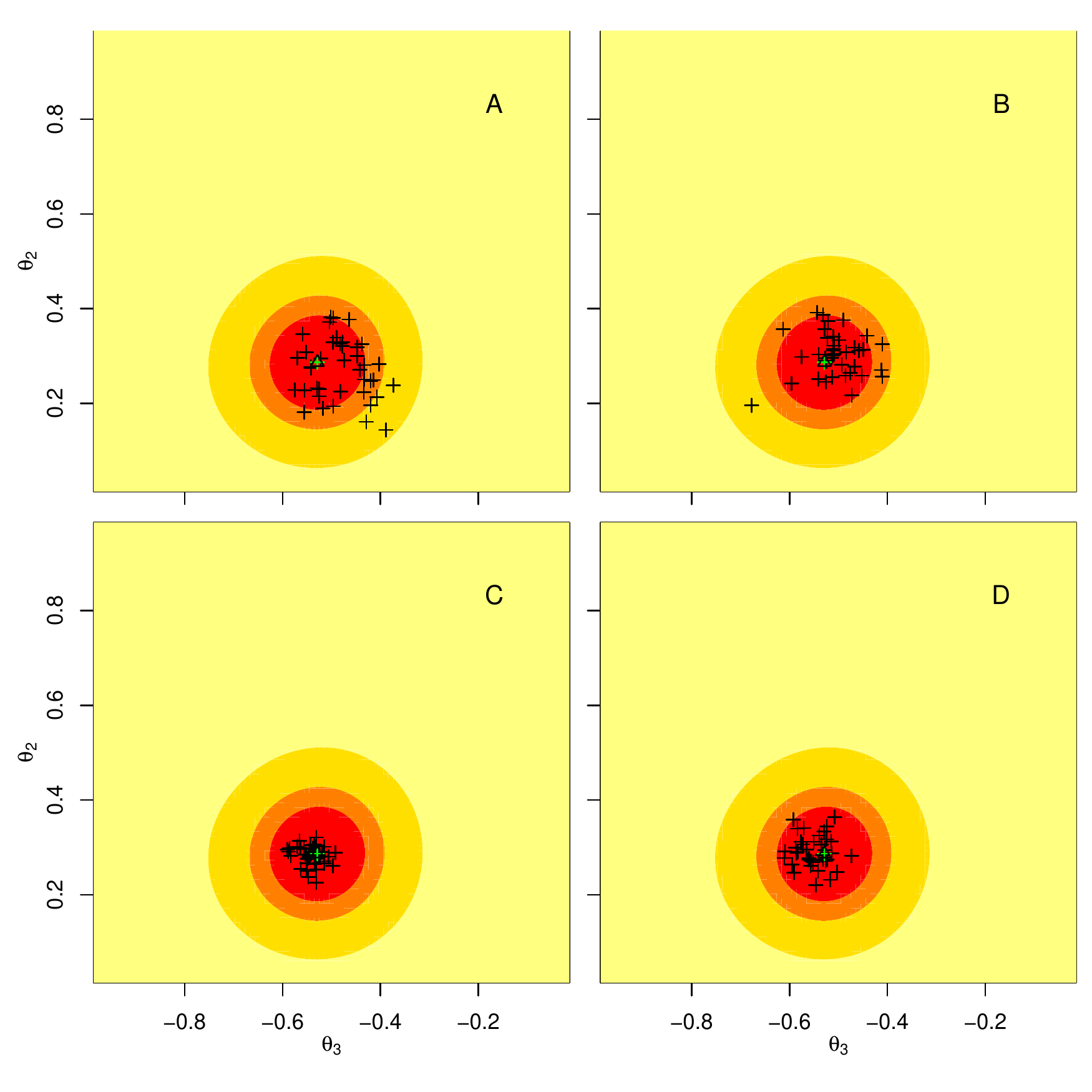} 
\caption{Comparison of different estimators. The likelihood surface for the
linear model, with the location of the MLE is marked with a green cross. 
The crosses show final points from 40 Monte
Carlo replications of the estimators: (A) Original iterated filtering method; (B)
Bayes map iterated filtering method; (C) Accelerate iterated filtering
method; (D) Second-order iterated filtering method; Each method, was
started uniformly over the rectangle shown, with $M=25$ iterations,
$N=1000$ particles, and a random walk standard deviation decreasing
from $0.02$ geometrically to $0.011$ for both $\alpha_{2}$ and
$\alpha_{3}$. }
\label{fig:sup:1} 
\end{figure}

To see how the final MLEs clustered around the true MLE, we only show
$40$ Monte Carlo replications for this toy example. As can be observed
from Fig.~\ref{fig:sup:1}, most of the replications clustered near
the true MLE for AIF approach, while none of them stays in a lower
likelihood region. It can be interpreted as a statistical summary
of Fig.~\ref{fig:sup:1}, with $200$ Monte Carlo replications. These
results indicate that AIF is clearly the best of the investigated
methods for this test compared to others. Given additional computational
resources, we also checked how the results of each method compared.
Specifically, we set $M=100$ iterations and $J=10000$ particles,
with the random walk standard deviation decreasing geometrically from
$0.02$ down to $0.0018$ for each method. In this situation, we 
confirm that AIF is the best among other IF1, IF2 and IS2. All methods
have comparable computational demands for given $M$ and $J$. 

\begin{figure}
\vspace{-0.5cm}
 \includegraphics[width=7cm,height=7cm]{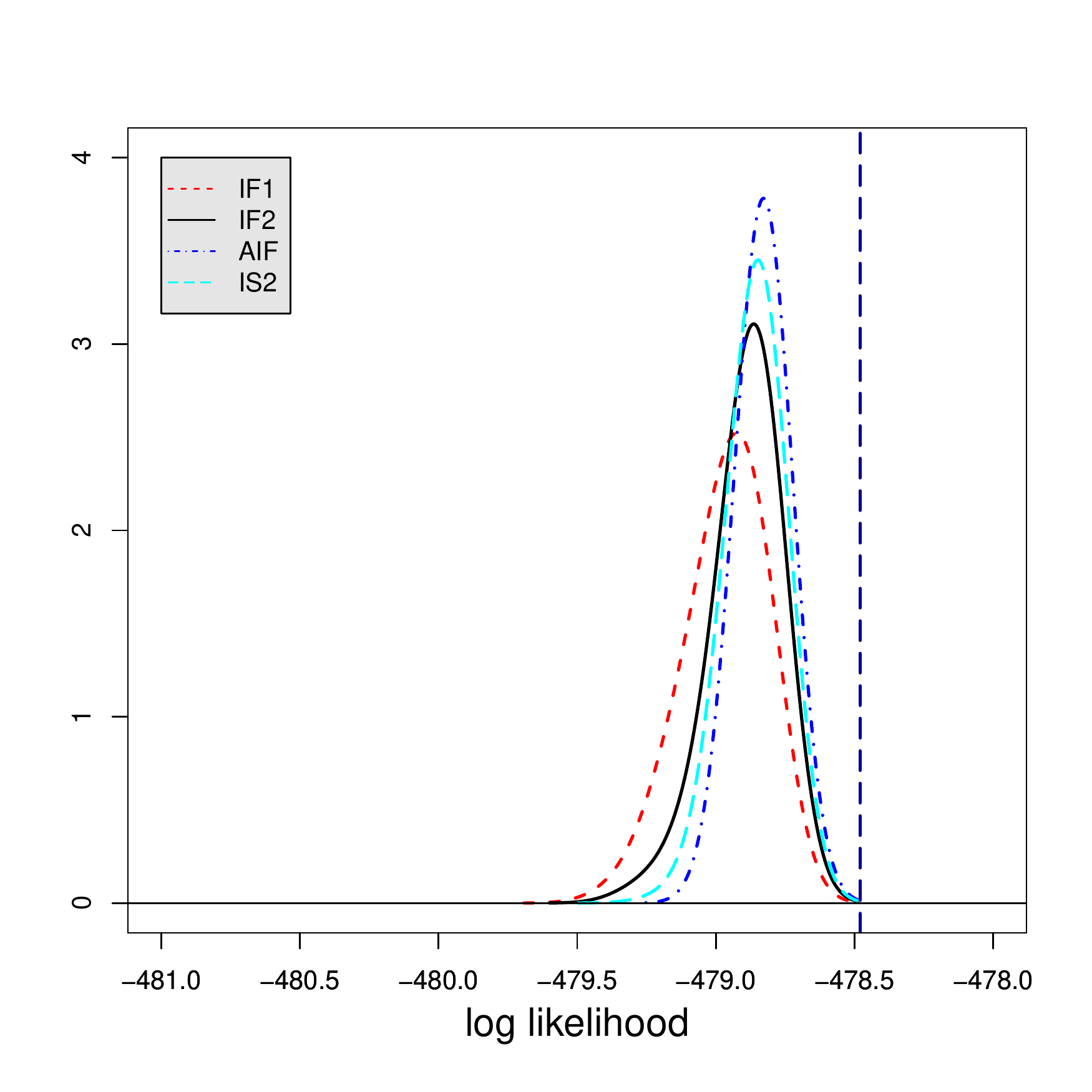} 
\caption{Comparison of estimators for the linear, Gaussian toy example, showing
the densities of the MLEs estimated by the IF1, IF2, AIF and
IS2 methods using $M=100$ iterations and $J=10000$ particles. The parameters $\alpha_{2}$ and $\alpha_{3}$ were
estimated, started from 200 randomly uniform initial values over a
large rectangular region $[-1,1]\times[-1,1]$. }
\label{fig:toy2} 
\end{figure}
\subsection{Malaria benchmark}
Many real world dynamic systems are highly nonlinear, partially observed
and even weakly identifiable. To demonstrate the capabilities of accelerate
iterated filtering for such situations, we apply it to evaluate the
likelihood in a stochastic differential equation for vivax malaria
model of \citet{roy12}. The reason to choose this challenging model
is that it provides a rigorous performance benchmark for our verification.
The model $SEIH^{3}QS$ we consider splits up the study population
of size $P(t)$ into seven classes: susceptible individuals, $S(t)$,
exposure $E(t)$, infected individuals, $I(t)$, dormant classes $H_{1}(t)$,
$H_{2}(t)$, $H_{3}(t)$ and recovered individuals, $Q(t)$. This
strain of malaria characterized by relapse following initial recovery
from symptoms \cite{nguyenis215}. Therefore the the last $S$ in
the model name indicates the possibility that a recovered person can
return to the class of susceptible individuals. The data, denoted
by $y_{1:N}^{*}$, are in the form of monthly time series over a
20-year period, counting the malaria morbidity. $\delta$ denotes the
mortality rate, $\kappa(t)$ a delay stage, $\mu_{SE}(t)$ the current
force of infection, and $\tau_{D}$ the mean latency time. The state process
is 
\[
X(t)=\big(S(t),E(t),I(t),Q(t),H_{1}(t),H_{2}(t),H_{3}(t),\kappa(t),\mu_{SE}(t)\big),
\]
where transition rates from stage $H_{1}$ to $H_{2}$, $H_{2}$ to
$H_{3}$ and $H_{3}$ to $Q$ are specified to be $3\mu_{HI}$ while
infected population to dormancy transition rate is $\mu_{IH}$.
The model satisfies the following stochastic differential equation
system
\begin{eqnarray*}
dS/dt & = & \delta P+\mathrm{d}P/dt+\mu_{IS}I+\mu_{QS}Q\\
 &  & \hspace{5mm}+a\mu_{IH}I+b\mu_{EI}E-\mu_{SE}(t)S-\delta S,\\
dE/dt & = & \mu_{SE}(t)S-\mu_{EI}E-\delta E,\\
dI/dt & = & (1-b)\mu_{EI}E+3\mu_{HI}H_{n}-(\mu_{IH}+\mu_{IS}+\mu_{IQ})I-\delta I,\\
dH_{1}/dt & = & (1-a)\mu_{IH}I-n\mu_{HI}H_{1}-\delta H_{1},\\
dH_{i}/dt & = & 3\mu_{HI}H_{i-1}-3\mu_{HI}H_{i}-\delta H_{i}\hspace{5mm}\mbox{for \ensuremath{i\in\{2,3\}}},\\
dQ/dt & = & \mu_{IQ}I-\mu_{QS}Q-\delta Q.
\end{eqnarray*}
In addition, the malaria pathogen reproduction within the mosquito
vector is given by 
\begin{eqnarray*}
\mathrm{d}\kappa/\mathrm{d}t & = & [\lambda(t)-\kappa(t)]/\tau_{D},\\
\mathrm{d}\mu_{SE}/\mathrm{d}t & = & [\kappa(t)-\mu_{SE}(t)]/\tau_{D},
\end{eqnarray*}
where $\lambda(t)$ is the latent force of infection and $\lambda(t)$,
$\kappa(t)$ and $\mu_{SE}(t)$ satisfies 
\begin{equation}
\mu_{SE}(t)=\int_{-\infty}^{t}\gamma(t-s)\lambda(s)\mathrm{d}s,
\end{equation}
with $\gamma(s)=\frac{(2/\tau_{D})^{2}s^{2-1}}{(2-1)!}\exp(-2s/\tau_{D})$,
a gamma distribution with shape parameter $2$. Since the latent force
of infection is constrained by rainfall covariate $R(t)$ and some
Gamma white noise, from \citet{roy12} we have:
\[
{\lambda(t)=\left(\frac{I+qQ}{P}\right)\times\exp\left\{ \sum_{i=1}^{N_{s}}b_{i}s_{i}(t)+b_{r}R(t)\right\} {\times\left[\frac{\mathrm{d}\Gamma(t)}{\mathrm{d}t}\right]}}.
\]
In this equation, $q$ denotes a reduced infection risk from humans
in the $\mathrm{Q}$ class and $\{s_{i}(t),i=1,\dots,N_{s}\}$ is
a periodic cubic B-spline basis, with $N_{s}=6$. The observation
model for $Y_{n}$ is a negative binomial distribution with mean $M_{n}$
and variance $M_{n}+M_{n}^{2}\sigma_{\mathrm{obs}}^{2}$ where $M_{n}=\rho\int_{t_{n-1}}^{t_{n}}[\mu_{EI}E(s)+3\mu_{HI}H_{3}(s)]ds$
is the number of new cases observed from time $t_{n-1}$to time $t_{n}$
and $\rho$ it the mean age. The coupled system of stochastic differential
equations is solved using an Euler-Maruyama scheme \citep{kloeden99}
with a time step of $1/20$ month in our case. 
\begin{figure}
\vspace{-0.5cm}
 \includegraphics[width=7cm,height=7cm]{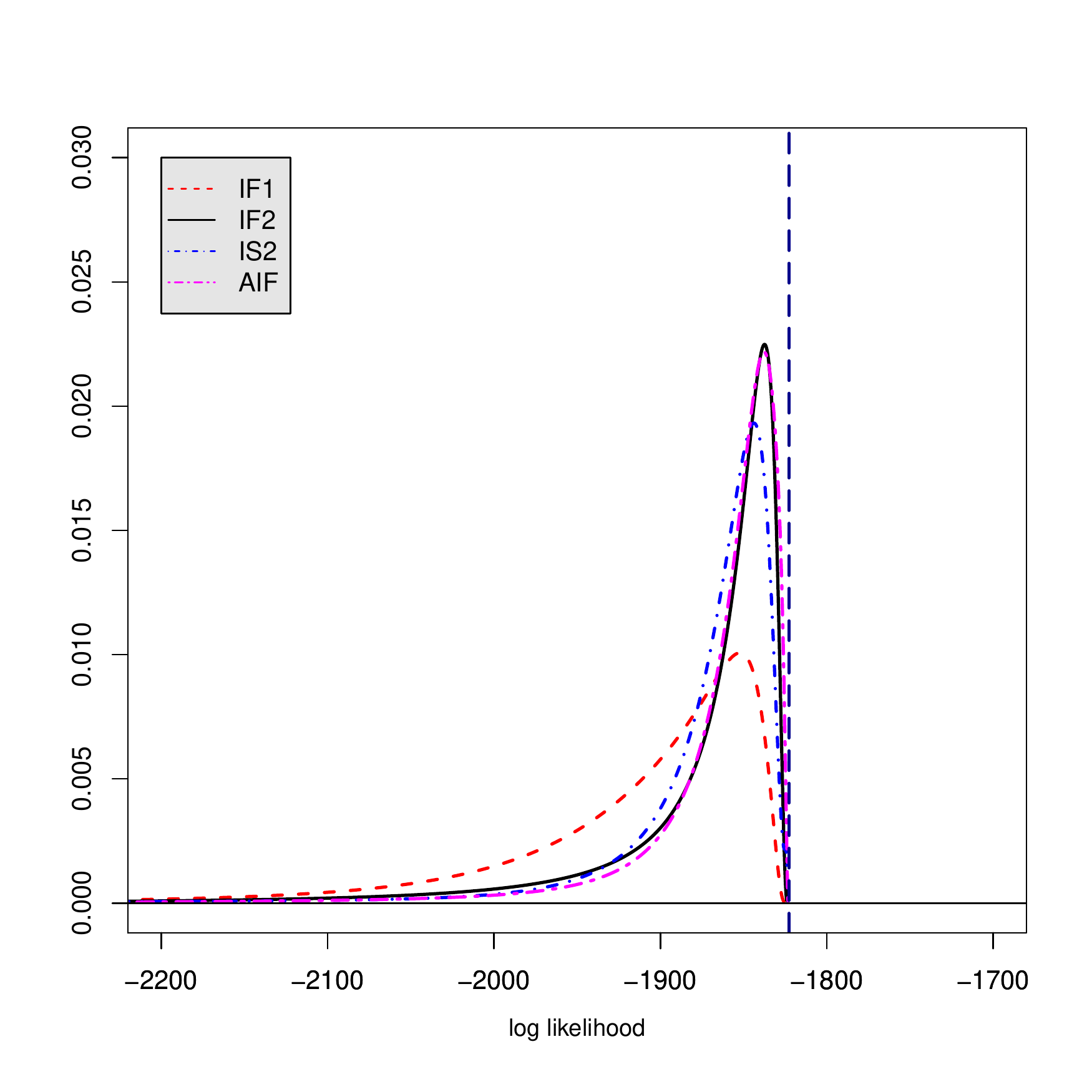} 
\caption{The density of the maximized log likelihood approximations
estimated by IF1, IF2, IS2, RIS1 and AIF for the malaria model when using
J = 1000 and M = 50. The log likelihood at a previously computed
MLE is shown as a dashed vertical line. }
\label{fig:malaria} 
\end{figure}
Given the data obtained from National Institutes of Malaria Research
\citep{roy12}, we carried out simulation-based inference via the
original iterated filtering (IF1), the perturbed Bayes map iterated filtering
(IF2), the second order iterated smoothing  (IS2), and the new accelerate
iterated filtering (AIF). The inference goal used to assess all of these methods is to find
high likelihood parameter values starting from randomly drawn values
in a large hyperrectangle. In the presence of possible
multi-modality, weak identifiability, and considerable Monte~Carlo
error of this model, we start $200$ random searches. The random walk
standard deviation is initially set to $0.1$ for estimated parameters
while the cooling rate $c$ is set to $0.1^{0.02}\approx0.95$. These
corresponding quantities for initial value parameters are $2$ and
$0.1^{0.02}$, respectively, but they are applied only at time zero.
We run our experiment on a cluster computers with $M=50$ iterations
and with $J=1000$ particles. The reason to choose these values for
this model is that increasing the iterations to $100$ and the number
of particles to $10000$ does not improve the results much but it
takes significant longer time. Figure~\ref{fig:malaria} shows the
distribution of the MLEs estimated by IF1, IF2, IS2 and AIF. All distributions
touch the global maximum as expected and the higher mean and smaller
variance of IF2, AIF estimation clearly demonstrate that they are
considerably more effective than IF1. Note that the computational
times for IF1, IF2, IS2 and AIF are 44.86, 43.92, 53.10 and 52.25 minutes respectively, confirming
that accelerate iterated filtering has essentially the same computational
cost as first order methods IF1, IF2 and is cheaper a bit than IS2, for a given Monte Carlo sample
size and number of iterations. In this hard problem, while IF1 reveal
their limitations, we have shown that IF2 and AIF can still offer
a substantial improvement. A natural heuristic idea to further improve
the method is hybridizing IF2 and AIF but we leave it for the future
work.  
 
\section{Conclusion\label{sec:5Conclusion}}

In this paper, we have proposed a novel class of iterated filtering theory using
an accelerated inexact gradient approach. We have shown that choosing 
perturbation sequence and number of particles carefully results in an 
algorithm which has led to many advances including the statistical and computational efficiency. This is also very
fruitful as it is extendable to a more generalized class of algorithm,
based on proximal theory. Previous proof of iterated filtering class
require some difficult conditions, which is not easily verifiable. However,
in this article, we use only general standard gradient conditions. We are going further down the road of
a more systematic approach which could be easily generalized to the state of the art
algorithm in the optimization literatures. The convergence rate is also explicitly stated 
and it is better than standard theory. From a theoretical 
point of view, it could be an interesting perspective and insight. 

In addition, from practical point of view, we have provided an efficient
framework, applicable to a general class of nonlinear, non-Gaussian
non-standard POMP models, especially suitable in the control feedback
system. There are a lot of such systems, which are not 
well-treated by current available modeling framework. We simultaneously
present the performance of our open source software package is2 to
facilitate the needs of the community. The performance of this new
approaches surpass the other frameworks by a large margin of magnitude.

It may be surprising that this simple accelerated inexact gradient approach has the
needed convergence properties, and can easily be generalized, at least
in some asymptotic sense. It is not hard to show that the accelerated inexact proximal gradient iterated filtering theory
can be adapted to apply with iterated smoothing and with either independent
white noise or random walk perturbations while our empirical results still show strong evidences of the improvements. 
In principle, 
different simulation-based inference methods can readily be hybridized
to build on the strongest features of multiple algorithms. Our results
could also be applied to develop other simulation-based methodologies
which can take advantage of proximal map.
For example, it may be possible to use our approach to help design efficient 
proposal distributions for particle Markov chain Monte Carlo algorithms.
The theoretical and algorithmic innovations of this paper will help to build a new direction
for future developments on this frontier. 
Applying this approach to methodologies like Approximate Bayesian Computation
(ABC), Liu-West Particle Filter (LW-PF), Particle Markov chain Monte Carlo (PMCMC), 
with different samplers scheme, 
e.g. forward backward particle filter, forward smoothing or forward backward smoothing
are foreseeable extensions.

\appendix
\section{Proofs \label{AppA}}

We first need a simple technical result (see Lemma 1 of \cite{ghadimi2016accelerated}).
The proof is the same as that of Lemma 1 of \cite{ghadimi2016accelerated}
but we provide it here for completeness.

\begin{lemma} (Lemma1 of \cite{ghadimi2016accelerated}. 

Assume sequences
$\{\alpha_{k}\}\in\left(0,1\right)$ for $k>1$ and $\alpha_{1}=1$
and sequences $\{a_{k}\},$ $\{\eta_{k}\}$ satisfy
\begin{equation}
a_{k}\leq(1\ -\alpha_{k})a_{k-1}+\eta_{k},\ k=1,2,\ldots\label{eq:2.5}
\end{equation}
If we define a positive sequence $\{\Gamma_{k}\}$ as in \ref{eq:2.6}
then for any $k\geq1$, we have 
$${\displaystyle a_{k}\leq\Gamma_{k}\sum_{i=1}^{k}(\eta_{i}/\Gamma_{i})}.$$
\end{lemma}
\begin{proof}
Since $\alpha_{1}=1$ and $\Gamma_{1}=1$, from \ref{eq:2.6}
we have 
\[
a_{1}\leq\eta_{1}
\]
or
\[
\frac{a_{1}}{\Gamma_{1}}\leq\frac{\eta_{1}}{\Gamma_{1}}.
\]
Since $\Gamma_{k}>0$ for every $k>1$, dividing both sides of \ref{eq:2.5}
by $\Gamma_{k}$, 
\[
\frac{a_{k}}{\Gamma_{k}}\leq\frac{(1-\alpha_{k})a_{k-1}+\eta_{k}}{\Gamma_{k}}=\frac{a_{k-1}}{\Gamma_{k-1}}+\frac{\eta_{k}}{\Gamma_{k}},\ \forall k\geq2.
\]
Summing up the above inequalities and rearranging the terms, the conclusion
follows.
\end{proof}
\begin{lemma} 
\begin{equation}
\sum_{\tau=1}^{k}\frac{\alpha_{\tau}}{\Gamma_{\tau}}=\frac{1}{\Gamma_{k}}.\label{eq:2.16}
\end{equation}
\end{lemma}
\begin{proof}
We have 
\[
\sum_{\tau=1}^{k}\frac{\alpha_{\tau}}{\Gamma_{\tau}}=\frac{\alpha_{1}}{\Gamma_{1}}+\sum_{\tau=2}^{k}\frac{1}{\Gamma_{\tau}}\left(1-(1-\alpha_{\tau})\right)
\]
\[
={\displaystyle \frac{1}{\Gamma_{1}}+\sum_{\tau=2}^{k}(\frac{1}{\Gamma_{\tau}}-\frac{1}{\Gamma_{\tau-1}})=\frac{1}{\Gamma_{k}}}.
\]
\end{proof}
\subsection{Proof of Theorem 3}
\begin{proof} The proof follows closely to the proof of theorem 1
of \cite{ghadimi2016accelerated} except we consider bias estimate
of the gradient. We first prove part a.

By \ref{eq:2.1} and \ref{eq:2.3}, we have 
\[
\ell(\theta_{k})\leq\ell(\theta_{k-1})+\left\langle \nabla\ell(\theta_{k-1}),\ \theta_{k}-\theta_{k-1}\right\rangle {\displaystyle +\frac{L}{2}\Vert\theta_{k}-\theta_{k-1}\Vert^{2}}
\]
\[
=\ell(\theta_{k-1})+\left\langle \left(\nabla\ell(\theta_{k-1})-\nabla\ell(\theta_{k}^{md})-\epsilon_{k}\right)+\left(\nabla\ell(\theta_{k}^{md})+\epsilon_{k}\right),\ -\lambda_{k}\left(\nabla\ell(\theta_{k}^{md})+\epsilon_{k}\right)\right\rangle 
\]
\[
+\frac{L\lambda_{k}^{2}}{2}\Vert\nabla\ell(\theta_{k}^{md})+\epsilon_{k}\Vert^{2}
\]
\[
=\ell(\theta_{k-1})-\lambda_{k}(1-\frac{L\lambda_{k}}{2})\Vert\nabla\ell\ (\theta_{k}^{md})+\epsilon_{k}\Vert^{2}-\lambda_{k}\left\langle \left(\nabla\ell(\theta_{k-1})-\nabla\ell(\theta_{k}^{md})-\epsilon_{k}\right),\ \left(\nabla\ell(\theta_{k}^{md})+\epsilon_{k}\right)\right\rangle 
\]
\[
\leq\ell(\theta_{k-1})-\lambda_{k}(1-\frac{L\lambda_{k}}{2})\Vert\nabla\ell\ (\theta_{k}^{md})+\epsilon_{k}\Vert^{2}+\lambda_{k}\left(\Vert\nabla\ell(\theta_{k-1})-\nabla\ell(\theta_{k}^{md})\Vert+\left\Vert \epsilon_{k}\right\Vert \right)\cdot\Vert\nabla\ell(\theta_{k}^{md})+\epsilon_{k}\Vert,
\]
\[
\leq\ell(\theta_{k-1})-\lambda_{k}(1-\frac{L\lambda_{k}}{2})\Vert\nabla\ell\ (\theta_{k}^{md})+\epsilon_{k}\Vert^{2}+\lambda_{k}\left(L\Vert\theta_{k-1}-\theta_{k}^{md}\Vert+\left\Vert \epsilon_{k}\right\Vert \right)\cdot\Vert\nabla\ell(\theta_{k}^{md})+\epsilon_{k}\Vert,
\]
\[
=\ell(\theta_{k-1})-\lambda_{k}(1-\frac{L\lambda_{k}}{2})\Vert\nabla\ell\ (\theta_{k}^{md})+\epsilon_{k}\Vert^{2}+\lambda_{k}\left(L(1-\alpha_{k})\Vert\theta_{k-1}^{ag}-\theta_{k-1}\Vert+\left\Vert \epsilon_{k}\right\Vert \right)\cdot\Vert\nabla\ell(\theta_{k}^{md})+\epsilon_{k}\Vert,
\]
\[
=\ell(\theta_{k-1})-\lambda_{k}\left(1-\frac{L\lambda_{k}}{2}\right)\left\Vert \nabla\ell(\theta_{k}^{md})+\epsilon_{k}\right\Vert {}^{2}
\]
\[
+L(1-\alpha_{k})\lambda_{k}\left\Vert \nabla\ell(\theta_{k}^{md})+\epsilon_{k}\right\Vert \cdot\left\Vert \theta_{k-1}^{ag}-\theta_{k-1}\right\Vert +\lambda_{k}\left\Vert \epsilon_{k}\right\Vert \cdot\left\Vert \nabla\ell(\theta_{k}^{md})+\epsilon_{k}\right\Vert 
\]
\[
\leq\ell\ (\theta_{k-1})-\lambda_{k}\left(1-\frac{L\lambda_{k}}{2}\right)\left\Vert \nabla\ell(\theta_{k}^{md})+\epsilon_{k}\right\Vert {}^{2}
\]
\[
+\frac{L\lambda_{k}^{2}}{2}\left\Vert \nabla\ell\ (\theta_{k}^{md})+\epsilon_{k}\right\Vert {}^{2}+\frac{L(1-\alpha_{k})^{2}}{2}\left\Vert \theta_{k-1}^{ag}-\theta_{k-1}\right\Vert {}^{2}+\lambda_{k}\left\Vert \epsilon_{k}\right\Vert \cdot\left\Vert \nabla\ell(\theta_{k}^{md})+\epsilon_{k}\right\Vert 
\]
\[
=\ell(\theta_{k-1})-\lambda_{k}(1-L\lambda_{k})\left\Vert \nabla\ell(\theta_{k}^{md})+\epsilon_{k}\right\Vert {}^{2}
\]
\begin{equation}
+{\displaystyle \frac{L(1-\alpha_{k})^{2}}{2}\left\Vert \theta_{k-1}^{ag}-\theta_{k-1}\right\Vert {}^{2}+\lambda_{k}\left\Vert \epsilon_{k}\right\Vert \cdot\left\Vert \nabla\ell(\theta_{k}^{md})+\epsilon_{k}\right\Vert }\label{eq:2.15}
\end{equation}
The second inequality is from triangular inequality and the Cauchy-Schwarz
inequality while the second inequality is due to the Lipschitz of
gradient assumption (1.2) and last equality comes from \ref{eq:2.2}.
We have the last inequality follows from $ab\leq(a^{2}+b^{2})/2$.
From \ref{eq:2.2}, \ref{eq:2.3}, and \ref{eq:2.4}, it follows that
\[
\theta_{k}^{ag}-\theta_{k}=(1-\alpha_{k})\theta_{k-1}^{ag}+\alpha_{k}\theta_{k-1}-\beta_{k}\left(\nabla\ell(\theta_{k}^{md})+\epsilon_{k}\right)-\left(\theta_{k-1}-\lambda_{k}\left(\nabla\ell(\theta_{k}^{md})+\epsilon_{k}\right)\right)
\]
\[
=(1-\alpha_{k})(\theta_{k-1}^{ag}-\theta_{k-1})+(\lambda_{k}-\beta_{k})\left(\nabla\ell(\theta_{k}^{md})+\epsilon_{k}\right).
\]
Applying Lemma 1 where $\theta_{k}^{ag}-\theta_{k}:=a_{k}$ and $\eta_{k}:=(\lambda_{k}-\beta_{k})\left(\nabla\ell(\theta_{k}^{md})+\epsilon_{k}\right)$,
we obtain 
\[
\theta_{k}^{ag}-\theta_{k}=\Gamma_{k}\sum_{\tau=1}^{k}(\frac{\lambda_{\tau}-\beta_{\tau}}{\Gamma_{\tau}})\left(\nabla\ell(\theta_{\tau}^{md})+\epsilon_{\tau}\right).
\]
Since$\left\Vert \cdot\right\Vert {}^{2}$ is convex, using Jensen's
inequality and Lemma 2 we have 
\[
\left\Vert \theta_{k}^{ag}-\theta_{k}\right\Vert {}^{2}=\left\Vert \Gamma_{k}\sum_{\tau=1}^{k}(\frac{\lambda_{\tau}-\beta_{\tau}}{\Gamma_{\tau}})\left(\nabla\ell\ (\theta_{\tau}^{md})+\epsilon_{k}\right)\right\Vert ^{2}
\]
\[
=\left\Vert \Gamma_{k}\sum_{\tau=1}^{k}\frac{\alpha_{\tau}}{\Gamma_{\tau}}\left[\left(\frac{\lambda_{\tau}-\beta_{\tau}}{\alpha_{\tau}}\right)\left(\nabla\ell(\theta_{\tau}^{md})+\epsilon_{k}\right)\right]\right\Vert ^{2}
\]
\[
\leq\Gamma_{k}\sum_{\tau=1}^{k}\frac{\alpha_{\tau}}{\Gamma_{\tau}}\left\Vert \left(\frac{\lambda_{\tau}-\beta_{\tau}}{\alpha_{\tau}}\right)\left(\nabla\ell(\theta_{\tau}^{md})+\epsilon_{k}\right)\right\Vert ^{2}
\]
\begin{equation}
={\displaystyle \Gamma_{k}\sum_{\tau=1}^{k}\frac{(\lambda_{\tau}-\beta_{\tau})^{2}}{\Gamma_{\tau}\alpha_{\tau}}\left\Vert \nabla\ell(\theta_{\tau}^{md})+\epsilon_{\tau}\right\Vert ^{2}}.\label{eq:2.17}
\end{equation}
Replacing the above bound in \ref{eq:2.15}, and the fact that $\Gamma_{k}=\Gamma_{k-1}(1-\alpha_{k})$
as in \ref{eq:2.6} and that $\alpha_{k}\in(0,\ 1]$ for all $k\geq1$we
obtain 
\[
\ell(\theta_{k})\leq\ell(\theta_{k-1})-\lambda_{k}(1-L\lambda_{k})\left\Vert \nabla\ell(\theta_{k}^{md})+\epsilon_{k}\right\Vert ^{2}
\]
\[
+\frac{L\Gamma_{k-1}(1-\alpha_{k})^{2}}{2}\sum_{\tau=1}^{k-1}\frac{(\lambda_{\tau}-\beta_{\tau})^{2}}{\Gamma_{\tau}\alpha_{\tau}}\left\Vert \nabla\ell(\theta_{\tau}^{md})+\epsilon_{\tau}\right\Vert ^{2}+\lambda_{k}\left\Vert \epsilon_{k}\right\Vert \cdot\left\Vert \nabla\ell(\theta_{k}^{md})+\epsilon_{k}\right\Vert 
\]
\[
\leq\ell(\theta_{k-1})-\lambda_{k}(1-L\lambda_{k})\left\Vert \nabla\ell(\theta_{k}^{md})+\epsilon_{k}\right\Vert ^{2}
\]
\begin{equation}
+{\displaystyle \frac{L\Gamma_{k}}{2}\sum_{\tau=1}^{k}\frac{(\lambda_{\tau}-\beta_{\tau})^{2}}{\Gamma_{\tau}\alpha_{\tau}}\left\Vert \nabla\ell(\theta_{\tau}^{md})+\epsilon_{\tau}\right\Vert ^{2}+\lambda_{k}\left\Vert \epsilon_{k}\right\Vert \cdot\left\Vert \nabla\ell(\theta_{k}^{md})+\epsilon_{k}\right\Vert }\label{eq:2.18}
\end{equation}
for every $k\geq1$. Using the definition of $C_{k}$ in \ref{eq:2.7}
and summing up the above inequalities, we have 
\[
\ell(\theta_{N})\ \leq\ell(\theta_{0})-\sum_{k=1}^{N}\lambda_{k}(1\ -L\lambda_{k})\left\Vert \nabla\ell(\theta_{k}^{md})+\epsilon_{k}\right\Vert ^{2}
\]
\[
+\frac{L}{2}\sum_{k=1}^{N}\Gamma_{k}\sum_{\tau=1}^{k}\frac{(\lambda_{\tau}-\beta_{\tau})^{2}}{\Gamma_{\tau}\alpha_{\tau}}\left\Vert \nabla\ell(\theta_{\tau}^{m})+\epsilon_{k}\right\Vert ^{2}+\sum_{k=1}^{N}\lambda_{k}\epsilon_{k}\cdot\left\Vert \nabla\ell(\theta_{k}^{md})+\epsilon_{k}\right\Vert 
\]
\[
=\ell(\theta_{0})-\sum_{k=1}^{N}\lambda_{k}(1\ -L\lambda_{k})\left\Vert \nabla\ell(\theta_{k}^{md})+\epsilon_{k}\right\Vert ^{2}
\]
\[
+\frac{L}{2}\sum_{k=1}^{N}\frac{(\lambda_{k}-\beta_{k})^{2}}{\Gamma_{k}\alpha_{k}}(\sum_{\tau=k}^{N}\Gamma_{\tau})\Vert\nabla\ell(\theta_{k}^{md})+\epsilon_{k}\Vert^{2}+\sum_{k=1}^{N}\lambda_{k}\left\Vert \epsilon_{k}\right\Vert \cdot\left\Vert \nabla\ell(\theta_{k}^{md})+\epsilon_{k}\right\Vert 
\]
\begin{equation}
=\ell(\theta_{0})-{\displaystyle \sum\lambda_{k}C_{k}\left\Vert \nabla\ell(\theta_{k}^{md})+\epsilon_{k}\right\Vert ^{2}}+\sum_{k=1}^{N}\lambda_{k}\left\Vert \epsilon_{k}\right\Vert \cdot\left\Vert \nabla\ell(\theta_{k}^{md})+\epsilon_{k}\right\Vert \label{eq:2.19}
\end{equation}
Rearranging the terms in the above inequality 
\[
{\displaystyle \sum_{k=1}^{N}\lambda_{k}C_{k}\left\Vert \nabla\ell(\theta_{k}^{md})+\epsilon_{k}\right\Vert ^{2}\leq\ell(\theta_{0})-\ell(\theta^*)}+\sum_{k=1}^{N}\lambda_{k}\left\Vert \epsilon_{k}\right\Vert \cdot\left\Vert \nabla\ell(\theta_{k}^{md})+\epsilon_{k}\right\Vert 
\]
By assumption \ref{ass5-1} that $\left\Vert \nabla\ell(\cdot)\right\Vert $
and $\sum_{k=1}^{N}\lambda_{k}\left\Vert \epsilon_{k}\right\Vert $
are bounded. Since $\ell(\theta_{N})\geq\ell(\theta^*)$ and in view of the
assumption that $C_{k}>0$, we obtain for some constant $B$,
\[
{\displaystyle \mathrm{min_{k=1,.N}}\Vert\nabla\ell(\theta_{k}^{md})+\epsilon_{k}\Vert^{2}\leq\frac{\ell(\theta_{0})-\ell(\theta^*)+B}{\sum_{k=1}^{N}\lambda_{k}C_{k}}}
\]
which clearly implies \ref{eq:2.8}.

We now prove part b).

First, from L-Lipschitz-continuous gradient property \ref{eq:2.4},
we have 
\[
\ell(\theta_{k}^{ag})\leq\ell(\theta_{k}^{md})+\left\langle \nabla\ell(\theta_{k}^{md}),\ \theta_{k}^{ag}-\theta_{k}^{md}\right\rangle +\frac{L}{2}\left\Vert \theta_{k}^{ag}-\theta_{k}^{md}\right\Vert ^{2}
\]
\begin{equation}
\leq\ell(\theta_{k}^{md})-{\displaystyle \beta_{k}\left\Vert \nabla\ell(\theta_{k}^{md})+\epsilon_{k}\right\Vert ^{2}+\left\Vert \epsilon_{k}\right\Vert \beta_{k}\left\Vert \nabla\ell(\theta_{k}^{md})+\epsilon_{k}\right\Vert +\frac{L\beta_{k}^{2}}{2}\left\Vert \nabla\ell(\theta_{k}^{md})+\epsilon_{k}\right\Vert ^{2}.}\label{eq:2.20}
\end{equation}
By the assumption that $\ell(\cdot)$ is convex and \ref{eq:2.2},
\[
\ell(\theta_{k}^{md})-\left[(1-\alpha_{k})\ell(\theta_{k-1}^{ag})+\alpha_{k}\ell(\theta)\right]
\]
\[
=\alpha_{k}\left[\ell(\theta_{k}^{md})-\ell(\theta)\right]+(1-\alpha_{k})\left[\ell(\theta_{k}^{md})-\ell(\theta_{k-1}^{ag})\right]
\]
\[
\leq\alpha_{k}\left\langle \nabla\ell(\theta_{k}^{md}),\ \theta_{k}^{md}-\theta\right\rangle +(1-\alpha_{k})\left\langle \nabla\ell(\theta_{k}^{md}),\ \theta_{k}^{md}-\theta_{k-1}^{ag}\right\rangle 
\]
\[
=\langle\nabla\ell(\theta_{k}^{md}),\ \alpha_{k}(\theta_{k}^{md}-\theta)\ +(1-\alpha_{k})(\theta_{k}^{md}-\theta_{k-1}^{ag})\rangle
\]
\begin{equation}
=\alpha_{k}\left\langle \nabla\ell(\theta_{k}^{md}),\ \theta_{k-1}-\theta\right\rangle .\label{eq:2.21}
\end{equation}
From \ref{eq:2.3}, we have 
\[
\left\Vert \theta_{k}-\theta\right\Vert ^{2}=\left\Vert \theta_{k-1}-\lambda_{k}\widehat{\nabla\ell(\theta_{k}^{md})}-\theta\right\Vert ^{2}
\]
\[
=\left\Vert \theta_{k-1}-\theta\right\Vert ^{2}-2\lambda_{k}\langle\widehat{\nabla\ell(\theta_{k}^{md})},\ \theta_{k-1}-\theta\rangle+\lambda_{k}^{2}\left\Vert \widehat{\nabla\ell(\theta_{k}^{md})}\right\Vert ^{2},
\]
\[
=\left\Vert \theta_{k-1}-\theta\right\Vert ^{2}-2\lambda_{k}\langle\nabla\ell(\theta_{k}^{md})+\epsilon_{k},\ \theta_{k-1}-\theta\rangle+\lambda_{k}^{2}\left\Vert \nabla\ell(\theta_{k}^{md})+\epsilon_{k}\right\Vert ^{2},
\]
which implies 
\[
\alpha_{k}\left\langle \nabla\ell(\theta_{k}^{md})+\epsilon_{k},\ \theta_{k-1}-\theta\right\rangle \ =\frac{\alpha_{k}}{2\lambda_{k}}\left[\left\Vert \theta_{k-1}-\theta\right\Vert ^{2}-\left\Vert \theta_{k}-\theta\right\Vert ^{2}\right]
\]
\[
+{\displaystyle \frac{\alpha_{k}\lambda_{k}}{2}\left\Vert \nabla\ell(\theta_{k}^{md})+\epsilon_{k}\right\Vert ^{2}}.
\]
Hence we obtain 
\[
\alpha_{k}\left\langle \nabla\ell(\theta_{k}^{md}),\ \theta_{k-1}-\theta\right\rangle \ \leq\frac{\alpha_{k}}{2\lambda_{k}}\left[\left\Vert \theta_{k-1}-\theta\right\Vert ^{2}-\left\Vert \theta_{k}-\theta\right\Vert ^{2}\right]
\]
\begin{equation}
+{\displaystyle \frac{\alpha_{k}\lambda_{k}}{2}\left\Vert \nabla\ell(\theta_{k}^{md})+\epsilon_{k}\right\Vert ^{2}}+\alpha_{k}\left\Vert \epsilon_{k}\right\Vert \left\Vert \theta_{k-1}-\theta\right\Vert \label{eq:2.22}
\end{equation}
Using the results of \ref{eq:2.20}, \ref{eq:2.21}, and \ref{eq:2.22},
we get 
\[
{\displaystyle \ell(\theta_{k}^{ag})\leq(1-\alpha_{k})\ell(\theta_{k-1}^{ag})+\alpha_{k}\ell(\theta)+\frac{\alpha_{k}}{2\lambda_{k}}\left[\left\Vert \theta_{k-1}-\theta\right\Vert ^{2}-\left\Vert \theta_{k}-\theta\right\Vert ^{2}\right]+\alpha_{k}\left\Vert \epsilon_{k}\right\Vert \left\Vert \theta_{k-1}-\theta\right\Vert }
\]
\[
-\beta_{k}(1-\frac{L\beta_{k}}{2}-\frac{\alpha_{k}\lambda_{k}}{2\beta_{k}})\left\Vert \nabla\ell(\theta_{k}^{md})+\epsilon_{k}\right\Vert ^{2}+\left\Vert \epsilon_{k}\right\Vert \beta_{k}\left\Vert \nabla\ell(\theta_{k}^{md})+\epsilon_{k}\right\Vert 
\]
\[
\leq\ (1\ -\alpha_{k})\ell(\theta_{k-1}^{ag})+\alpha_{k}\ell(\theta)+\frac{\alpha_{k}}{2\lambda_{k}}\left[\left\Vert \theta_{k-1}-\theta\right\Vert ^{2}-\left\Vert \theta_{k}-\theta\right\Vert ^{2}\right]
\]
\begin{equation}
-{\displaystyle \frac{\beta_{k}}{2}(1\ -L\beta_{k})\left\Vert \nabla\ell(\theta_{k}^{md})+\epsilon_{k}\right\Vert ^{2}}+\left\Vert \epsilon_{k}\right\Vert \beta_{k}\left\Vert \nabla\ell(\theta_{k}^{md})+\epsilon_{k}\right\Vert +\alpha_{k}\left\Vert \epsilon_{k}\right\Vert \left\Vert \theta_{k-1}-\theta\right\Vert ,\label{eq:2.23}
\end{equation}
where the last inequality follows from the assumption in \ref{eq:2.9}.
Subtracting $\ell(\theta)$ from both sides of the above inequality
and using Lemma 1, we conclude that 
\[
{\displaystyle \ell(\theta_{N}^{ag})-\ell(\theta)\leq\Gamma_{N}\left[\sum_{k=1}^{N}\frac{\alpha_{k}}{2\lambda_{k}\Gamma_{k}}\left[\left\Vert \theta_{k-1}-\theta\right\Vert ^{2}-\left\Vert \theta_{k}-\theta\right\Vert ^{2}\right]\right.}
\]
\[
\left.-\sum_{k=1}^{N}\frac{\beta_{k}}{2\Gamma_{k}}(1-L\beta_{k})\left\Vert \nabla\ell(\theta_{k}^{md})+\epsilon_{k}\right\Vert ^{2}+\sum_{k=1}^{N}\frac{1}{\Gamma_{k}}\left[\left\Vert \epsilon_{k}\right\Vert \beta_{k}\left\Vert \nabla\ell(\theta_{k}^{md})+\epsilon_{k}\right\Vert +\alpha_{k}\left\Vert \epsilon_{k}\right\Vert \left\Vert \theta_{k-1}-\theta\right\Vert \right]\right]
\]
\[
\leq\Gamma_{N}\frac{\left\Vert \theta_{0}-\theta\right\Vert ^{2}}{2\lambda_{1}}-\Gamma_{N}\sum_{k=1}^{N}\frac{\beta_{k}}{2\Gamma_{k}}\ (1\ -L\beta_{k})\left\Vert \nabla\ell(\theta_{k}^{md})+\epsilon_{k}\right\Vert ^{2}
\]
\begin{equation}
+\Gamma_{N}\sum_{k=1}^{N}\frac{1}{\Gamma_{k}}\left[\left\Vert \epsilon_{k}\right\Vert \beta_{k}\left\Vert \nabla\ell(\theta_{k}^{md})+\epsilon_{k}\right\Vert +\alpha_{k}\left\Vert \epsilon_{k}\right\Vert \left\Vert \theta_{k-1}-\theta\right\Vert \right]\label{eq:2.24}
\end{equation}
for every $\theta\in\mathbb{R}^{n}.$ By our contruction \ref{eq:2.10}
that sequence $\left\{ \frac{\alpha_{k}}{\lambda_{k}\Gamma_{k}}\right\} $is
decreasing and the fact that $\alpha_{1}=\Gamma_{1}=1$, we have 
\begin{equation}
{\displaystyle \sum_{k=1}^{N}\frac{\alpha_{k}}{\lambda_{k}\Gamma_{k}}\left[\left\Vert \theta_{k-1}-\theta\right\Vert ^{2}-\left\Vert \theta_{k}-\theta\right\Vert ^{2}\right]\leq\frac{\alpha_{1}\left\Vert \theta_{0}-\theta\right\Vert ^{2}}{\lambda_{1}\Gamma_{1}}=\frac{\left\Vert \theta_{0}-\theta\right\Vert ^{2}}{\lambda_{1}}}\label{eq:2.25}
\end{equation}
which immediately implies the last inequality of \ref{eq:2.24}.

Hence, we can conclude \ref{eq:2.12} from the above inequality and
the assumption in \ref{eq:2.9}: 
\[
{\displaystyle \ell(\theta_{N}^{ag})-\ell(\theta^{*})\leq\Gamma_{N}\left[\frac{\left\Vert \theta_{0}-\theta^{*}\right\Vert ^{2}}{\lambda_{1}}+\sum_{k=1}^{N}\Gamma_{k}^{-1}\left[\left\Vert \epsilon_{k}\right\Vert \beta_{k}\left\Vert \nabla\ell(\theta_{k}^{md})+\epsilon_{k}\right\Vert +\alpha_{k}\left\Vert \epsilon_{k}\right\Vert \left\Vert \theta_{k-1}-\theta\right\Vert \right]\right]}
\]
Finally, noting the fact that $\ell(\theta_{N}^{ag})\geq\ell(\theta^{*})$,
substitute $\theta:=\theta^{*}$, re-arranging the terms in \ref{eq:2.24}
we obtain 
\[
{\displaystyle \sum_{k=1}^{N}\frac{\beta_{k}}{2\Gamma_{k}}(1\ -L\beta_{k})\left\Vert \nabla\ell(\theta_{k}^{md})+\epsilon_{k}\right\Vert ^{2}\mbox{\,\ }k=1,\ldots,N}
\]
\[
\leq\frac{\Vert\theta^{*}-\theta_{0}\Vert^{2}}{2\lambda_{1}}+\sum_{k=1}^{N}\frac{1}{\Gamma_{k}}\left[\left\Vert \epsilon_{k}\right\Vert \beta_{k}\left\Vert \nabla\ell(\theta_{k}^{md})+\epsilon_{k}\right\Vert +\alpha_{k}\left\Vert \epsilon_{k}\right\Vert \left\Vert \theta_{k-1}-\theta\right\Vert \right],
\]
or 
\[
\mathrm{min_{k=1,.N}}\Vert\nabla\ell(\theta_{k}^{md})+\epsilon_{k}{\displaystyle \Vert^{2}\leq2\frac{\frac{\Vert\theta^{*}-\theta_{0}\Vert^{2}}{2\lambda_{1}}+\sum_{k=1}^{N}\frac{1}{\Gamma_{k}}\left[\left\Vert \epsilon_{k}\right\Vert \beta_{k}\left\Vert \nabla\ell(\theta_{k}^{md})+\epsilon_{k}\right\Vert +\alpha_{k}\left\Vert \epsilon_{k}\right\Vert \left\Vert \theta_{k-1}-\theta\right\Vert \right]}{\sum_{k=1}^{N}\Gamma_{k}^{-1}\beta_{k}(1-L\beta_{k})}}
\]
which together with \ref{eq:2.9}, clearly imply \ref{eq:2.11}. \end{proof}

\subsection{Proof of Theorem 4}

\begin{proof} We first prove part a). Note that by choosing 
\[
\beta_{k}=\frac{1}{2L}
\]
\begin{equation}
{\displaystyle {\displaystyle \Gamma_{k}=\frac{1}{k^{1+\delta}}}},\label{eq:2.33}
\end{equation}
which implies that for sufficient large $k$ 
\[
\sum_{\tau=k}^{N}\Gamma_{\tau}=\sum_{\tau=k}^{N}\frac{1}{\tau^{1+\delta}}=O(\frac{1}{k^{\delta}})
\]
We also have 
\begin{equation}
1-\alpha_{k}=\frac{(k-1)^{1+\delta}}{k^{1+\delta}}\label{eq:2.34}
\end{equation}
for every $k>1,$ or $\alpha_{k}=\frac{\left(k^{1+\delta}-\left(k-1\right)^{1+\delta}\right)}{k^{1+\delta}}=O(\frac{\left(1+\delta\right)k^{\delta}}{k^{1+\delta}})=O(\frac{1}{k})$.
If we choose $\lambda_{k}$ such that $\lambda_{k}-\beta_{k}=o(k^{-1})$
then 
\[
\frac{(\lambda_{k}-\beta_{k})^{2}}{2\alpha_{k}\Gamma_{k}\lambda_{k}}(\sum_{\tau=k}^{N}\Gamma_{\tau})=\frac{o(k^{-2})}{k^{-1}k^{-(1+\delta)}}\frac{1}{k^{\delta}}=o(1)
\]
so for sufficiently large $k$ we have 
\[
C_{k}=1-L[\lambda_{k}+\frac{(\lambda_{k}-\beta_{k})^{2}}{2\alpha_{k}\Gamma_{k}\lambda_{k}}(\sum_{\tau=k}^{N}\Gamma_{\tau})]>\frac{1}{4}
\]
Hence, it can also be seen from \ref{eq:2.8} that for some positive
bounded constant $B_{2}$, 
\[
{\displaystyle \mathrm{min_{k=1,.N}}\left\Vert \nabla\ell(\theta_{k}^{md})+\epsilon_{k}\right\Vert ^{2}\leq\frac{\ell(\theta_{0})-\ell^{*}+B}{NB_{2}}=O(\frac{1}{N})},
\]
which concludes the first part of the proof. Since $\left\Vert \epsilon_{k}\right\Vert =O\left(\tau^{2}\right)\leq O(\frac{1}{k})$,
we have $\nabla\ell(\theta_{k}^{md})$ converge to $0$ at the rate
of 
\[
\min\left\{ O(\frac{1}{\sqrt{N}}),O\left(\left\Vert \epsilon_{k}\right\Vert \right)\right\} =O(\frac{1}{\sqrt{N}}),
\]
which gives us the desired result.

We now show part b). Let $\lambda_{k}=\left(k^{1+\delta}-\left(k-1\right)^{1+\delta}\right)c$
for some constant $c$ then 
\[
\frac{\alpha_{1}}{\lambda_{1}\Gamma_{1}}=\frac{\alpha_{2}}{\lambda_{2}\Gamma_{2}}=\cdots=\frac{\alpha_{k}}{\lambda_{k}\Gamma_{k}}.
\]
Observe that 
\[
\alpha_{k}\lambda_{k}=\frac{c^{2}\left(k^{1+\delta}-\left(k-1\right)^{1+\delta}\right)^{2}}{k^{1+\delta}}=\frac{c^{2}\left(1+\delta\right)^{2}O(k^{2\delta})}{k^{1+\delta}}\rightarrow0
\]
for $\delta<1$ so $\frac{\left(1+\delta\right)^{2}k^{2\delta}}{k^{1+\delta}}<\beta_{k}=\frac{1}{2L}$
for sufficient large $k$, which implies that conditions \ref{eq:2.9}
and \ref{eq:2.10} hold. Moreover, it can also be easily seen from
\ref{eq:2.15} that 
\[
\mathrm{min_{k=1,.N}}\Vert\nabla\ell(\theta_{k}^{md})+\epsilon_{k}{\displaystyle \Vert^{2}\leq\frac{\frac{\Vert\theta^{*}-\theta_{0}\Vert^{2}}{2\lambda_{1}}+C\sum_{k=1}^{N}\Gamma_{k}^{-1}\left[\left\Vert \epsilon_{k}\right\Vert +O(\frac{1}{k})\left\Vert \epsilon_{k}\right\Vert \right]}{\sum_{k=1}^{N}\Gamma_{k}^{-1}}}=O(N^{-2-\delta}).
\]
The last equality is due to the fact that $\sum_{k=1}^{N}\Gamma_{k}^{-1}=\sum_{k=1}^{N}k^{(1+\delta)}=O(N^{2+\delta})$.
Combining the above relation with \ref{eq:2.8}, and since $\left\Vert \epsilon_{k}\right\Vert =O\left(\tau^{2}\right)\leq O(\frac{1}{k^{2+\delta+\delta_{1}}})$
for some $\delta_{1}>0$, we have $\nabla\ell(\theta_{k}^{md})$ converge
to $0$ at the rate of $O\left(\sqrt{\frac{1}{N^{2+\delta}}}\right).$

Since $\alpha_{k}\lambda_{k}<\beta_{k}=\frac{1}{2L}$, we have $\delta\leq1$
which implies that the best convergence rate is $O(N^{-3})$. \end{proof}


\begin{thebibliography}{}

\bibitem[Andrieu et~al., 2010]{andrieu10}
Andrieu, C., Doucet, A., and Holenstein, R. (2010).
\newblock Particle {M}arkov chain {M}onte {C}arlo methods.
\newblock {\em Journal of the Royal Statistical Society: Series B (Statistical
  Methodology)}, 72(3):269--342.

\bibitem[Andrieu and Vihola, 2015]{andrieu2015convergence}
Andrieu, C. and Vihola, M. (2015).
\newblock Convergence properties of pseudo-marginal {Markov chain Monte Carlo}
  algorithms.
\newblock {\em The Annals of Applied Probability}, 25(2):1030--1077.

\bibitem[Bhadra, 2010]{bhadra10}
Bhadra, A. (2010).
\newblock Discussion of `particle {M}arkov chain {M}onte {C}arlo methods' by
  {C}.\ {A}ndrieu, {A}.\ {D}oucet and {R}.\ {H}olenstein.
\newblock {\em Journal of the Royal Statistical Society, Series B (Statistical
  Methodology)}, 72:314--315.

\bibitem[Bret{\'o}, 2014]{breto2014idiosyncratic}
Bret{\'o}, C. (2014).
\newblock On idiosyncratic stochasticity of financial leverage effects.
\newblock {\em Statistics \& Probability Letters}, 91:20--26.

\bibitem[Bret\'{o} and Ionides, 2011]{breto11}
Bret\'{o}, C. and Ionides, E.~L. (2011).
\newblock Compound {M}arkov counting processes and their applications to
  modeling infinitesimally over-dispersed systems.
\newblock {\em Stoch. Process. Their Appl.}, 121(11):2571--2591.

\bibitem[Dahlin et~al., 2015]{DahlinLindstenSchon2015a}
Dahlin, J., Lindsten, F., and Sch\"{o}n, T.~B. (2015).
\newblock {Particle Metropolis-Hastings using gradient and Hessian
  information}.
\newblock {\em Statistics and Computing}, 25(1):81--92.

\bibitem[Doucet et~al., 2002]{doucet02}
Doucet, A., Godsill, S.~J., and Robert, C.~P. (2002).
\newblock Marginal maximum a posteriori estimation using {M}arkov chain {M}onte
  {C}arlo.
\newblock {\em Statistics and Computing}, 12:77--84.

\bibitem[Doucet et~al., 2013]{doucet2013derivative}
Doucet, A., Jacob, P.~E., and Rubenthaler, S. (2013).
\newblock Derivative-free estimation of the score vector and observed
  information matrix with application to state-space models.
\newblock {\em ArXiv:1304.5768}.

\bibitem[Eddelbuettel et~al., 2011]{eddelbuettel2011rcpp}
Eddelbuettel, D., Fran{\c{c}}ois, R., Allaire, J., Ushey, K., Kou, Q., Russel,
  N., Chambers, J., and Bates, D. (2011).
\newblock Rcpp: Seamless r and c++ integration.
\newblock {\em Journal of Statistical Software}, 40(8):1--18.

\bibitem[Gaetan and Yao, 2003]{gaetan03}
Gaetan, C. and Yao, J.-F. (2003).
\newblock A multiple-imputation {M}etropolis version of the {EM} algorithm.
\newblock {\em Biometrika}, 90(3):643--654.

\bibitem[Ghadimi and Lan, 2016]{ghadimi2016accelerated}
Ghadimi, S. and Lan, G. (2016).
\newblock Accelerated gradient methods for nonconvex nonlinear and stochastic
  programming.
\newblock {\em Mathematical Programming}, 156(1-2):59--99.

\bibitem[Girolami and Calderhead, 2011]{girolami2011riemann}
Girolami, M. and Calderhead, B. (2011).
\newblock {Riemann manifold Langevin and Hamiltonian Monte Carlo methods}.
\newblock {\em Journal of the Royal Statistical Society: Series B (Statistical
  Methodology)}, 73(2):123--214.

\bibitem[Ionides et~al., 2011]{ionides11}
Ionides, E.~L., Bhadra, A., Atchad{\'e}, Y., and King, A. (2011).
\newblock Iterated filtering.
\newblock {\em Annals of Statistics}, 39:1776--1802.

\bibitem[Ionides et~al., 2006]{ionides06-pnas}
Ionides, E.~L., Bret{\'o}, C., and King, A.~A. (2006).
\newblock Inference for nonlinear dynamical systems.
\newblock {\em Proceedings of the National Academy of Sciences of the USA},
  103:18438--18443.

\bibitem[Ionides et~al., 2015]{ionides15}
Ionides, E.~L., Nguyen, D., Atchad{\'e}, Y., Stoev, S., and King, A.~A. (2015).
\newblock Inference for dynamic and latent variable models via iterated,
  perturbed {B}ayes maps.
\newblock {\em Proceedings of the National Academy of Sciences of the USA},
  112(3):719--724.

\bibitem[Jacquier et~al., 2007]{jacquier07}
Jacquier, E., Johannes, M., and Polson, N. (2007).
\newblock {MCMC} maximum likelihood for latent state models.
\newblock {\em Journal of Econometrics}, 137(2):615--640.

\bibitem[King et~al., 2016]{king15pomp}
King, A.~A., Nguyen, D., and Ionides, E.~L. (2016).
\newblock Statistical inference for partially observed {M}arkov processes via
  the {R} package pomp.
\newblock {\em Journal of Statistical Software}, 69(12).

\bibitem[Kloeden and Platen, 1999]{kloeden99}
Kloeden, P.~E. and Platen, E. (1999).
\newblock {\em Numerical Soluion of Stochastic Differential Equations}.
\newblock Springer, New York, 3rd edition.

\bibitem[Kushner and Clark, 1978]{kushner78}
Kushner, H.~J. and Clark, D.~S. (1978).
\newblock {\em Stochastic Approximation Methods for Constrained and
  Unconstrained Systems}.
\newblock Springer-Verlag, New York.

\bibitem[Laneri et~al., 2010]{laneri10}
Laneri, K., Bhadra, A., Ionides, E.~L., Bouma, M., Dhiman, R.~C., Yadav, R.~S.,
  and Pascual, M. (2010).
\newblock Forcing versus feedback: Epidemic malaria and monsoon rains in
  {N}orthwest {I}ndia.
\newblock {\em PLoS Computational Biology}, 6(9):e1000898.

\bibitem[Leggetter and Woodland, 1995]{leggetter1995maximum}
Leggetter, C.~J. and Woodland, P.~C. (1995).
\newblock Maximum likelihood linear regression for speaker adaptation of
  continuous density hidden {M}arkov models.
\newblock {\em Computer Speech \& Language}, 9(2):171--185.

\bibitem[Lele et~al., 2007]{lele07}
Lele, S.~R., Dennis, B., and Lutscher, F. (2007).
\newblock Data cloning: Easy maximum likelihood estimation for complex
  ecological models using {B}ayesian {M}arkov chain {M}onte {C}arlo methods.
\newblock {\em Ecology Letters}, 10(7):551--563.

\bibitem[Lindstr{\"o}m, 2013]{lindstrom2013tuned}
Lindstr{\"o}m, E. (2013).
\newblock Tuned iterated filtering.
\newblock {\em Statistics \& Probability Letters}, 83(9):2077--2080.

\bibitem[Lindstr\"{o}m et~al., 2012]{lindstrom12}
Lindstr\"{o}m, E., Ionides, E.~L., Frydendall, J., and Madsen, H. (2012).
\newblock Efficient iterated filtering.
\newblock In {\em 16th IFAC Symposium on System Identification}.

\bibitem[Nemeth et~al., 2013]{nemeth2013particle}
Nemeth, C., Fearnhead, P., and Mihaylova, L. (2013).
\newblock Particle approximations of the score and observed information matrix
  for parameter estimation in state space models with linear computational
  cost.
\newblock {\em ArXiv:1306.0735}.

\bibitem[Nemeth et~al., 2014]{nemeth2014sequential}
Nemeth, C., Fearnhead, P., and Mihaylova, L. (2014).
\newblock Sequential monte carlo methods for state and parameter estimation in
  abruptly changing environments.
\newblock {\em IEEE Transactions on Signal Processing}, 62(5):1245--1255.

\bibitem[Nesterov, 2005]{nesterov2005smooth}
Nesterov, Y. (2005).
\newblock Smooth minimization of non-smooth functions.
\newblock {\em Mathematical Programming}, 103(1):127--152.

\bibitem[Nesterov, 2013]{nesterov2013introductory}
Nesterov, Y. (2013).
\newblock {\em Introductory Lectures on Convex Optimization: A Basic Course},
  volume~87.
\newblock Springer Science \& Business Media.

\bibitem[Nguyen, 2016]{nguyen2016another}
Nguyen, D. (2016).
\newblock {Another look at Bayes map iterated filtering}.
\newblock {\em Statistics \& Probability Letters}.

\bibitem[Nguyen and Ionides, 2017]{nguyenis215}
Nguyen, D. and Ionides, E.~L. (2017).
\newblock A second-order iterated smoothing algorithm.
\newblock {\em Statistical Computing}.

\bibitem[Pitt et~al., 2012]{pitt2012some}
Pitt, M.~K., dos Santos~Silva, R., Giordani, P., and Kohn, R. (2012).
\newblock {On some properties of Markov chain Monte Carlo simulation methods
  based on the particle filter}.
\newblock {\em Journal of Econometrics}, 171(2):134--151.

\bibitem[Poyiadjis et~al., 2011]{Poyiadjis-etal:2009}
Poyiadjis, G., Doucet, A., and Singh, S.~S. (2011).
\newblock Particle approximations of the score and observed information matrix
  in state space models with application to parameter estimation.
\newblock {\em Biometrika}, 98(1):65--80.

\bibitem[{R Core Team}, 2013]{amanual}
{R Core Team} (2013).
\newblock {\em R: A Language and Environment for Statistical Computing}.
\newblock R Foundation for Statistical Computing, Vienna, Austria.

\bibitem[Roy et~al., 2013]{roy12}
Roy, M., Bouma, M.~J., Ionides, E.~L., Dhiman, R.~C., and Pascual, M. (2013).
\newblock The potential elimination of plasmodium vivax malaria by relapse
  treatment: {I}nsights from a transmission model and surveillance data from
  {NW} {I}ndia.
\newblock {\em PLoS Neglected Tropical Diseases}, 7(1):e1979.

\end{thebibliography}


%%%%%%%%%%%%%%%%%%%%%%%%%%%%%%%%%%%%%%%%%%%%%%%%%%%%%%%%%
%%% acknownledgements
%%%%%%%%%%%%%%%%%%%%%%%%%%%%%%%%%%%%%%%%%%%%%%%%%%%%%%%%%

\section*{Acknowledgements}
This research was funded in part by National Science Foundation grant DMS-1308919.

\bibliographystyle{apalike}

\end{document}